\documentclass[journal]{IEEEtran}
\IEEEoverridecommandlockouts
\usepackage{cite}
\usepackage{subcaption}
\usepackage[shortlabels]{enumitem}
\usepackage{graphicx}
\usepackage{textcomp}
\usepackage{xcolor}
\usepackage{cite}
\usepackage{amsmath,amssymb,amsfonts}
\usepackage{relsize}
\usepackage{algorithm,algorithmic}
\usepackage{float}
\usepackage{enumitem}  
\usepackage{amsmath}
\usepackage{color}
\usepackage{soul}
\usepackage{footnote}
\usepackage{mathtools}
\DeclarePairedDelimiter\ceil{\lceil}{\rceil}

\usepackage[T1]{fontenc}
\def\BibTeX{{\rm B\kern-.05em{\sc i\kern-.025em b}\kern-.08em
    T\kern-.1667em\lower.7ex\hbox{E}\kern-.125emX}}

\newcommand\Tau{\mathrm{T}}

\usepackage{epstopdf}
\usepackage{dsfont}

\usepackage{amsmath,epsfig,amssymb,amsthm,cite,url}
\usepackage{verbatim}
\usepackage[english]{babel}
\usepackage{amsmath,amssymb}

\usepackage{layouts}

\newtheorem{definition}{Definition}
\newtheorem{theorem}{Theorem}

\usepackage[cmintegrals]{newtxmath}
\renewcommand{\openbox}{\leavevmode
  \hbox to.77778em{%
  \hfil\vrule
  \vbox to.675em{\hrule width.6em\vfil\hrule}%
  \vrule\hfil}}
    
\usepackage{needspace}
\begin{document}


\title{Smart City Defense Game: Strategic Resource Management during Socio-Cyber-Physical Attacks}

\author{Dimitrios~Sikeridis,~\IEEEmembership{Student~Member,~IEEE}, and Michael Devetsikiotis,~\IEEEmembership{Fellow,~IEEE}

\thanks{D. Sikeridis, and M. Devetsikiotis are with the Department of Electrical and Computer Engineering, The University of New Mexico, Albuquerque, NM, USA, (e-mail: dsike@unm.edu, mdevets@unm.edu).}
\thanks{The research of D. Sikeridis, and M. Devetsikiotis was supported in part by the US National Science Foundation under the NM EPSCoR cooperative agreement Grant OIA-1757207. 
We acknowledge the fruitful discussions with UNM ECE's PROTON Lab group under Prof. E.E. Tsiropoulou during early stages of this work.}
}

\maketitle

\begin{abstract}
Ensuring public safety in a Smart City (SC) environment is a critical and increasingly complicated task due to the involvement of multiple agencies and the city's expansion across cyber and social layers. In this paper, we propose an extensive form game with complete information and observed actions to model interactions and optimal city resource allocations when a Terrorist Organization (TO) performs attacks on multiple targets across two conceptual SC levels, a physical, and a cyber-social (e.g., via AI-generated disinformation). The Smart City Defense Game (SCDG) considers three players that initially are entitled to a specific finite budget. Two SC agencies that have to defend their physical or social territories respectively, fight against a common enemy, the TO. Each layer consists of multiple targets and the attack outcome depends on whether the resources allocated there by the associated agency, exceed or not the TO's. Each player's utility is equal to the total value of the successfully defended (or attacked) targets. The two agencies are allowed to make budget transfers provided that neither is worse off. We completely characterize the Sub-game Perfect Nash Equilibrium (SPNE) of the SCDG that consists of strategies for optimal resource exchanges between SC agencies and accounts for the TO's budget allocation across the physical and social targets. Also, we present numerical and comparative results demonstrating that when the SC players act according to the SPNE, they maintain or increase the number of successfully defended targets in comparison to the case where no transfer takes place. The SCDG is shown to be a promising solution for modeling critical resource allocations between SC parties in the face of multi-layer simultaneous terrorist attacks.
\end{abstract}


\section{Introduction}

\looseness=-1 The Smart City (SC) paradigm aims for enhanced citizen life quality and safety by incorporating innovative applications that rely on diverse technologies (e.g. Internet of Things (IoT), cloud computing, big data analytics, and artificial intelligence (AI)) \cite{nam2011smart, khatoun2016smart, eckhoff2017privacy}. Arguably, the intelligence of such SC environments also stems from their ability to make decisions related to the use and management of their natural and municipal resources, both in the short term and when accounting for future development \cite{pettit2018planning}. This is not an easy task, especially since the SC organization incorporates a set of primal city entities with specific resource budgets, separate governance structures, and unique operational goals. Such smart city entities primarily include traffic and public transport authorities, departments overlooking critical cyber-physical facilities (i.e., intelligent buildings \cite{sikeridis2018unsupervised}, smart grid, natural gas, or water infrastructures \cite{ferdowsi2017colonel}), information and communication technology (ICT) administrations, and public safety/emergency service agencies (ESAs)\cite{eckhoff2017privacy}. Since these entities often operate on different conceptual levels (physical, cyber, social) within the city structure, the SC paradigm requires a management platform for supervision, coordination and optimal strategic allocation of SC resources, especially in cases of public safety threatening events such as adversarial/human-caused attacks. A case in point is Rio de Janeiro's SC operation center that integrates multiple individual agencies towards optimal disaster response and emergency management \cite{singer2012mission}.

\looseness=-1 In response, high-level city adversaries on the physical plane like traditional terrorist organizations have advanced their tactics towards inflicting the maximum possible damage by distributing their forces across multiple city targets. Several terrorist efforts attest to this observation with the Paris attacks in 2015 taking place simultaneously across six distinct physical locations \cite{hirsch2015medical}, and the Brussels bombings in 2016 occurring in coordination across two different city targets \cite{stieglitz2018sense}. Thus, terrorist strategies can be guided by knowledge related to the city authorities' structure such as the distribution of first response resources. This knowledge can be acquired by practical means including adversarial insiders, social engineering against city officers/employees (e.g., by social-media data exploitation), and long-term extraction/analysis of SC open data.

\looseness=-1 In addition, amidst the era of social media (SM) and generative AI, terrorist groups are exploiting technological advancements and trends to improve their tactics \cite{burke2016age, glazzard2026generative}. This adaptation creates new city vulnerabilities for exploitation, especially since social media are considered a cyber-social extension of the future SC. Interestingly, in the last decade, the increasing adoption of SM by citizens and ICT city agencies during emergencies creates a propagation of information to many directions\cite{reuter2018fifteen}. This includes citizen to citizen (self-organization, alerting and aid), SC ICT agencies to citizens (and traditional media to citizens - for public alerting and guidance), and citizens to ICT agencies (SM integration into monitoring environments for intelligence extraction, situation awareness, and immediate response)\cite{stieglitz2018sense, reuter2018fifteen, imran2015processing, xu2017social}. Examples of the massive use of SM during terrorist attacks include the Brussels bombings in 2016 \cite{stieglitz2018sense}, and the Boston Marathon bombings in 2013 \cite{starbird2014rumors, zhang2018scalable}, where a major concern regarding the credibility of posted information emerged. Specifically, during this incident, the diffusion of misinformation and speculation through Twitter actively endangered individual lives and led to misuse of emergency response resources\cite{starbird2014rumors, zhang2018scalable}. This emerging dependency on SM during crises can be exploited by sophisticated SC adversaries to produce disinformation streams, including AI-generated text and deepfakes \cite{glazzard2026generative}, towards directing the public to unsafe city zones or actively obstructing the operation of SC ESAs. To this end, next-generation terrorist organizations can make use of massive social bots \cite{wang2018era} and AI swarms, i.e., coordinated AI agents that adapt in real time with minimal human oversight \cite{schroeder2026malicious}, to generate targeted SM posts, or partner with hacker communities to infiltrate SC alerting/ICT infrastructures \cite{roygame}. Large language models (LLMs) facilitate such campaigns by generating convincing falsehoods inexpensively and at scale \cite{schroeder2026malicious, park2026generative}.

In this work we consider a multi-layer smart city model and present a defense mechanism for optimal SC resource allocation in response to simultaneous terrorist attacks of various types. The key contributions of our research work are summarized as follows:

a) A Smart City is modeled as a multi-dimensional setting which consists of a lower physical plane and an upper cyber-social one. The physical layer is a set of distinct SC physical locations and city points of interest, while the social layer consists of multiple cyber spaces that include social media, web pages and chat spaces.

\looseness=-1 b) By considering a Terrorist Organization (TO) attack taking place in both SC layers, we model the optimal response of two SC agencies responsible for public safety and SC defense, namely an Emergency Service Agency (ESA) operating at the physical layer and an Information and Communication Technology (ICT) agency operating at the Cyber-Social SC plane. Each organization aims to deploy its financial resources optimally across multiple spaces of interest either physical or cyber, by also considering the possibility of budget exchange with the other agency. We also take into account the TO financial strength and respective budget allocation across SC targets. In order to fully capture the inter-dependencies and interactions among all the conflicting parties we introduce a multi-stage Smart City Defense Game (SCDG) with observed actions and compute the game's subgame perfect Nash equilibrium that describes the optimal strategies of all players focusing on the optimal budget exchange among the SC agencies that reduces the expected number of successful TO attacks across the two SC layers and targets. The equilibrium is given in closed form for all budget relations and game values, and is formally verified by the Lean proof assistant.

c)  Detailed numerical and comparative results demonstrate that the proposed Smart City Defense Game is a promising solution for modeling SC agencies' resource allocations, internal budget transfers and interactions with a conflicting party towards securing the cyber-physical Smart City of the future. 

\section{Related Work}

\looseness=-1 Previous research in the general area of anti-terrorism conflicts utilizes game theory (security games) to model interaction between a defender and a sophisticated attacker \cite{sandler2003terrorism}. The common assumption in these works is this of the passive defender that allocates resources before any attack without actively harming his adversary. In this context, two cases can be distinguished: (a) the attacker is not aware of the defender's play, and (b) he has complete knowledge. In the first case, the majority of games are simultaneous yielding Nash equilibria in mixed strategies where both players randomize over their actions to confuse their opponents. The most popular game in this category is the Colonel Blotto game, and its numerous variations, where the players allocate finite resources over multiple battlefields \cite{roberson2006colonel}. Similarly the work in \cite{arce2012weakest} considers a one-to-one resource allocation game across multiple locations with a mixed-strategy equilibrium where the terrorist can choose between attacking with a suicide bomber or use conventional force. In \cite{baron2018game}, the authors develop mixed strategies for the two opponents that choose between two actions for each target, namely act (attack or defend) or not. The authors investigate single or multiple-period security games to examine ongoing conflicts where the terrorist can use one or several attack technologies with different capabilities.  

\looseness=-1 However, game-theoretic frameworks should account for the fact that terrorist strategies can adapt in response to the defender's actions, while the interactions between the two opponents are interdependent. Thus, the second category of research works assumes complete knowledge among the two parties and utilizes multi-stage games that are solved using backward induction yielding sub-game perfect Nash equilibria (SPNEs). In \cite{berman2007location} the authors model a two-stage Stackelberg game (leader-follower game) where the state initially decides where to locate resource-packed facilities and the terrorist, given these locations, decides on the attack targets that will maximize the inflicted damage. The work in \cite{meng2018determining} extends this model to account for disruption in the defender's facilities with a non-zero probability of failures on the supply-side (resource unavailability) and propose a heuristic algorithm to solve the developed problem. 

\looseness=-1 In our study, a mixed approach is followed where a stage of resource allocation/preparation (pure strategies) precedes the actual allocation among targets modeled as a Colonel Blotto game (mixed strategies). In addition, our model considers two defenders and the creation of a coalition among them. In a similar fashion, the authors in \cite{carceles2011strategic} present a multi-stage sequential game model in which a set of different countries are confronted by an international terrorist organization. Initially, countries invest resources to fight proactively the adversary ($1^{st}$ stage), next all countries allocate defense resources ($2^{nd}$ stage), and finally, the terrorist allocates attacking resources among countries ($3^{rd}$ stage). The game studies the nation's cooperation against the terrorist, and yields an SPNE, while in contrast to our work the defensive measures of a specific country can direct terrorist attacks to other allies (the cooperation is not always beneficial). The opposite case of collusive behavior among attackers is studied by Roy et al. in \cite{roygame}. Their work introduces a coalition formation game that investigates the characteristics and the mechanisms of alliance creation among terrorist/hacker organizations against a single defender. Finally, the work in \cite{gholami2016divide} models the case of multiple adversaries against a single defender as a Stackelberg security game (defender $\equiv$ leader,  attackers $\equiv$ followers), and calculates the optimal defense strategy given knowledge of payoff matrices, and target-related attack-success probabilities.

\begin{figure*}[t]
\centering
\includegraphics[width=\textwidth]{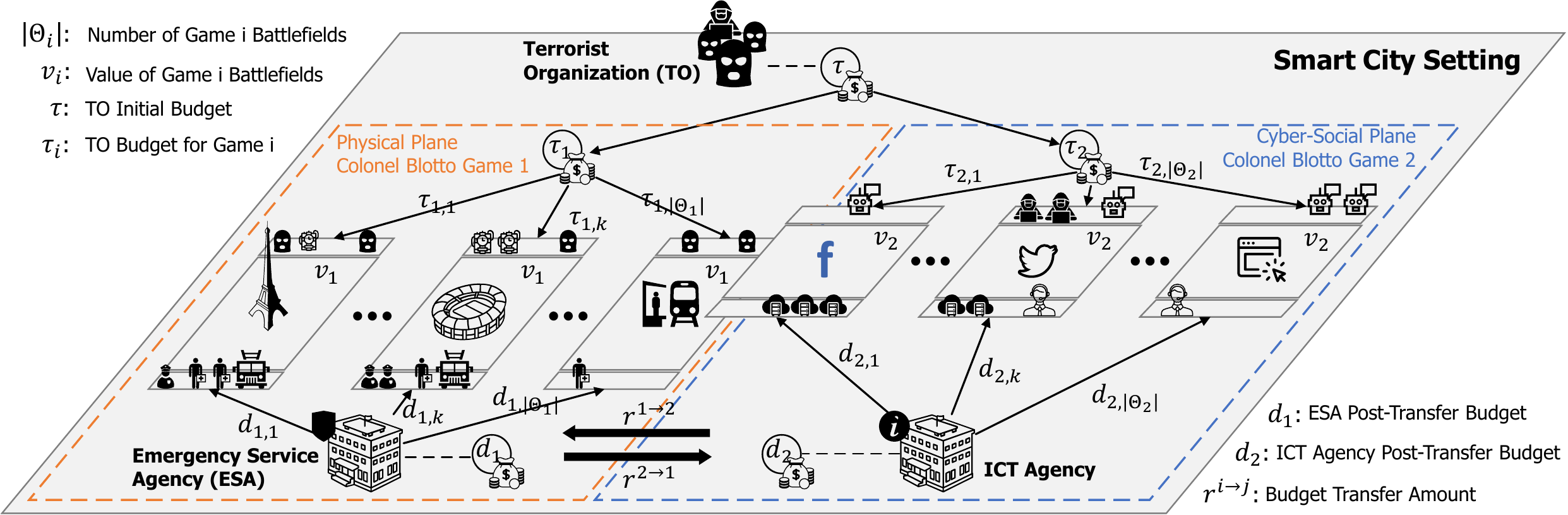}
\caption{\small Smart City Defense Game: Players, Games, and Components }  \label{fig:SysArch}
\vspace{-0.4cm}
\end{figure*}

\looseness=-1 Other research works that focus on smart city security do so by considering a cyber-physical system perspective with the adversary attempting to compromise individual components. In~\cite{ferdowsi2020interdependence} the authors develop a game theoretic framework to defend intelligent transportation systems against indirect attacks carried out through the power grid. In \cite{ferdowsi2017colonelSC} the authors model CPS elements of a smart city as connected nodes using graph theory and develop a Colonel Blotto-based resource allocation game between a defender and an attacker that tries to compromise the nodes. Similarly, the work in \cite{ferdowsi2017colonel} examines an SC's interdependent critical infrastructure (ICI) consisting of power-gas-water distribution systems and considers a two-stage attack to a set of ICI sensor's protection and state estimation quality. The attacker-defender interaction is modeled as a Colonel Blotto allocation game where the SC administrator allocates resources with the form of computational, communication or financial resources to establish protection levels for the ICI nodes. The authors derive a Mixed Strategy Nash Equilibrium (MSNE) for the two players and examine the optimal defender's strategy for a series of different cases. In \cite{gupta2014three} the authors formulate a multi-stage Blotto game where a single adversary fights against two defenders. The game focuses on the cyber vulnerability of servers against a hacker, has a hierarchical structure similar to our proposed work, and each defender has to decide whether or not to add additional battlefields to the CB games, or transfer resources to the other player. Other coalitional Blotto games consider three defenders that share resources \cite{heyman2018colonel}, transfers where part of the budget is lost \cite{shah2024inefficient}, exchanges of both budget and battlefields \cite{shah2024battlefields}, and the comparison of such alliance mechanisms \cite{shah2026alliance}. The works in \cite{kovenock2012coalitional, shah2024inefficient, shah2024battlefields, shah2026alliance} focus on when beneficial transfers exist, while in our study we derive the SPNE transfers of the two SC agencies for all budget relations. Finally, the work in \cite{rao2016defense}  examines attacks on the cyber and physical parts of a wide-area network testbed. The attacker chooses to attack a single part, while the defender chooses -or not- to reinforce the whole infrastructure acting according to a game-theoretic framework that yields a pure Nash equilibrium for the two opponents.

\Needspace{6\baselineskip}
\section{Smart City Defense Game Model and\\ Problem Formulation}

In this section we set the stage for the proposed Smart City security game which is formulated as a complete information multi-stage game with three players.

\subsection{Attack and Defense Scenarios}

\looseness=-1 Consider a smart city adversary, namely a Terrorist Organization (TO) $\Tau$ with finite available resources represented by a financial budget $\tau$. The organization deploys a parallel attack towards the SC targeting simultaneously two separate conceptual levels. At the first level, the TO uses a part of its financial resources, denoted as $\tau_{1}$ to perform physical attacks on multiple critical SC area targets by allocating attack budget to each site that can be translated to human agents (suicide vehicles, bombers, shooters, etc.) or attacking equipment. In response, the SC's emergency service agency (ESA) deploys its own financial resources denoted as $c_{1}$ across the critical targeted areas for defense and disaster prevention purposes. The ESA's budget can be translated to first responder units (human resources, police, firefighters, medical personnel) and emergency management equipment. 

\looseness=-1 In addition, in our model, the sophisticated terrorists attack concurrently a second "cyber-social" level of the SC environment by using another part of their resource budget $\tau_{2}$ with $\tau_{1}+\tau_{2}=\tau$. This is achieved by allocating the $\tau_{2}$ attack budget to disseminate disinformation across multiple social media, traditional media sites or SC alerting infrastructure to either obscure the truth and affect the general public or tamper with social sensing applications utilized by the SC entities \cite{stieglitz2018sense, imran2015processing, xu2017social, zhang2018scalable}. The cyber-social attack budget is utilized by the adversary either towards securing computational resources to enable autonomous AI-driven social-bot and AI swarm operation for disinformation diffusion \cite{lopezjoya2025dissecting, wang2018era, schroeder2026malicious} or for acquiring human resources responsible for the same task (partner with hacker organizations under hire etc.\cite{roygame}). On the SC defense side, the ICT administration which is responsible for securing the information-related SC layers utilizes its pre-allocated defense budget $c_{2}$ by allocating it across the different social media entities under attack. The ICT financial resources can be used either 
\begin{enumerate}[a)]
    \item  for deploying ICT administration human resources responsible for identifying/exposing unreliable sources and AI-generated content, and providing trustworthy news to the public, or
    \item  for dynamically securing and acquiring cloud computing resources (usually offered by public cloud service providers \cite{sikeridis2017comparative}, similar to the case of IBM in Rio \cite{singer2012mission}). Such resources (computing power for real-time data analytics, machine learning frameworks \cite{sikeridis2017comparative}, and LLMs \cite{park2026generative}) can be utilized for deploying truth discovery algorithms that identify misinformation in the presence of noisy data from unvetted SM sources (e.g., as in \cite{zhang2018scalable} where the proposed solution was evaluated against real-world Twitter datasets extracted from past terrorist attacks).
\end{enumerate}

\looseness=-1 For both scenarios, we will assume that the party that has allocated the majority of resources in each targeted area (either physical or cyber-social) has successfully achieved his goal (landed a successful attack or managed to defend the target). Since the satisfaction of each player depends not only on his actions but also on the actions of his opponent (i.e., the number of resources strategically allocated) we can use game theory to model their interactions \cite{fudenberg1991game}. Thus, in order to model (a) the player interactions on the two parallel city levels (physical and cyber-social) and (b) model their allocation of budget across multiple city area targets and across multiple cyber-social spaces (e.g. different social media), we assume that the TO participates concurrently to two Colonel Blotto games \cite{roberson2006colonel} against the two city entities. We will further assume that the two SC entities are able to form a coalition towards exchanging emergency resources if neither of them is worse off. In order to model this resource transfer and examine its characteristics given the TO's own allocation of resources among his two rivals (i.e., the two  SC entities and in extension the two SC layers) we formulate a multi-stage complete information Smart City Defense Game (SCDG) which is partly based on the multi-stage Blotto game described in \cite{kovenock2012coalitional}. In what follows we define and describe the basic parts of the game.

\subsection{The Colonel Blotto Game}\label{sub:CBG}

\looseness=-1 The continuous Colonel Blotto Game \cite{roberson2006colonel} models the strategic resource allocation between two opponents with finite infinitely divisible resources (troops) in a competitive environment that consists of multiple battlefields. The two opponents play the game by allocating their troops to each battlefield. The player that allocated the larger amount wins the battlefield while their objective is to win as many battles as possible. It is a one-shot game defined as CBG$\{P,\{F^{p}\}_{p \in P},\{S^{p}\}_{p\in P}, \Theta, v, \{ U^{p}\}_{p \in P} \}$ where:
\begin{itemize}
    \item $P \triangleq \{P_{A},P_{B}\}$ denotes the two opponents/players
    \item $F^{p}$ are the available resources of player $p \in P$
    \item $S^{p}$ is the set of strategies for player $p$, $p \in P$
    \item $\Theta$ is the set of the game's battlefields with $\theta = |\Theta|$
    \item $v$ denotes the value of each battlefield
    \item $U^{p}$ is the utility function of player $p$, $p \in P$
\end{itemize}
The two players distribute their total forces $F^{p}$ across the $\theta$ battlefields with the allocation vector of player $p$ being $\boldsymbol{f}^{p}=[f^{p}_{1},...,f^{p}_{k},...,f^{p}_{\theta}]$, where $f^{p}_{k}$ is the resource amount assigned to battlefield $k$. Thus, the strategies of each player form the set $S^{p}$ of all the possible allocations across the battlefields:
\begin{equation*}\small
    S^{p}  \triangleq \{ \boldsymbol{f}^{p}~ |~ \sum_{k=1}^{\theta} f^{p}_{k} \leq F^{p},~~f^{p}_{k}\geq 0 \}
\end{equation*}
Each battlefield is won by the player with the highest resource contribution, while the payoff of player $p$ from winning a single battlefield $k$ is defined as:
\begin{equation*}\small
    u^{p}_{k}({f}^{p}_{k},{f}^{-p}_{k}) = \begin{cases}
                                        v &\text{if ${f}^{p}_{k}>{f}^{-p}_{k}$}\\
                                        0 &\text{if ${f}^{p}_{k}<{f}^{-p}_{k}$}\\
                                        \frac{v}{2} &\text{if ${f}^{p}_{k}={f}^{-p}_{k}$}
\end{cases}
\end{equation*}
where ${f}^{-p}_{k}$ denotes the opponent's resource contribution to battlefield $k$. The opponent's payoff per battlefield is $u^{-p}_{k}({f}^{p}_{k},{f}^{-p}_{k})=v - u^{p}_{k}({f}^{p}_{k},{f}^{-p}_{k})$. The overall utility of each player is defined as:
\begin{equation*}\small
    U^{p}({\boldsymbol{f}}^{p},{\boldsymbol{f}}^{-p})= \sum_{k=1}^{\theta}u^{p}_{k}({f}^{p}_{k},{f}^{-p}_{k})
\end{equation*}
The goal of each player $p$ is to choose a strategy in $S^{p}$ (i.e., a resource allocation vector) that maximizes his utility and number of won battlefields given his opponent's selected strategy. 
\begin{definition}\label{definition}
For the CBG a strategy profile $\{\boldsymbol{f}^{p*},\boldsymbol{f}^{-p*}\}$, $\boldsymbol{f}^{p*} \in S^{p}$ and  $\boldsymbol{f}^{-p*} \in S^{-p}$ is a pure-strategy Nash equilibrium if for every player $p \in P$:
\begin{equation}\small
    U^{p}(\boldsymbol{f}^{p*},\boldsymbol{f}^{-p*})\geq U^{p}(\boldsymbol{f}^{p},\boldsymbol{f}^{-p*}),~ \forall \boldsymbol{f}^{p} \in S^{p}.
\end{equation}
\end{definition}
It has been proven in \cite{roberson2006colonel} that the CBG is not guaranteed to yield an NE in pure-strategies. Therefore, an NE for the CBG exists in mixed-strategies, where each opponent $p \in P$ chooses a multivariate probability distribution function over $S^{p}$ (assigns a probability for playing each pure strategy). A CBG mixed strategy for player $p$ is a distribution of resources expressed by a $\theta$-variate distribution function $O^{p}:\mathbb{R}^{\theta}_{+} \xrightarrow{} [0,1]$ with support contained inside the set $S^{p}$ of feasible resource allocations. 
\begin{definition}\label{definition2}
Let $\mathcal{Q}^{p}$ be the set of all probability distributions over player's $p$ pure-strategy space $S^{p}$. For the CBG a mixed strategy profile set $\{O^{p*},O^{-p*}\}$ is a mixed-strategy Nash equilibrium (MSNE) if for every player $p \in P$:
\begin{equation}\small
    U^{p}(O^{p*},O^{-p*})\geq U^{p}(O^{p},O^{-p*}),~ \forall O^{p} \in \mathcal{Q}^{p}.
\end{equation}
\end{definition}
\noindent
Each $\theta$-variate distribution function $O^{p}$ is associated with a set of univariate
marginal distribution functions $\{ \Phi_{k}^{p}\}_{k=1}^{\theta}:\mathbb{R}_{+} \xrightarrow{} [0,1]$ for each battlefield $k$. For a player $p$, given his mixed strategy NE, the forces' allocation vector $\boldsymbol{f}^{p}=[f^{p}_{1},...,f^{p}_{k},...,f^{p}_{\theta}]$ is drawn from $O^{p}$ with $f^{p}_{k}$ being a random variable drawn from $\Phi_{k}^{p}$.

 \subsection{The Smart City Defense Game}

\looseness=-1 Given the SC attack/defense scenarios and the CBG discussion above we formulate a multi-stage SC Defense Game (SCDG) with observed actions that consists of three players and captures all interactions between allies and opponents. The two SC entities, namely the Emergency Service Agency and the ICT agency, that will be denoted as player 1 and 2 respectively, fight against the Terrorist Organization denoted by T. The pre-allocated financial defense budget of each SC player $i \in \{1,2\}$ is $c_{i}$, while the Terrorist organization's attack budget is $\tau$. The two-layer conflict takes place simultaneously across $\theta_{1}$ SC area physical targets (set $\Theta_{1}$) that yield a payoff of $v_{1}$ to the winner (TO or ESA agency), and across $\theta_{2}$ social media/cyber-social targets (set $\Theta_{2}$) that yield a payoff of $v_{2}$ to the winner (TO or ICT agency) assuming $\theta_{i} \geq 3$ $\forall i\in\{1,2\}$ and $\theta_{i} \in \mathbb{N}$, $v_{i} \in \mathbb{R}_{>0}$.

\looseness=-1 The SCDG is an extensive form game with complete information and observed actions, whose model parameters and actions taken by all players during previous stages are common knowledge. Thus, at the beginning of each stage there is a well-defined history $h_{stage}$, and a set of all possible histories $H_{stage}$. For this initial first stage $h_{1}=\emptyset$, and $\Pi=\{c_{1},c_{2},\tau,v_{1},v_{2},\theta_{1},\theta_{2}\}$ is the set of initial SCDG parameters that describe the setting. During the first stage the two SC entities form a coalition and choose whether to make a budget transfer towards their ally or not while the TO performs no action. We denote the amount of financial resource transfer from SC agency $i$ to agency $j$ as $r^{i\rightarrow{} j} \in [0,c_{i})$, while $\{r^{1\rightarrow{} 2},  r^{2\rightarrow{} 1} \}$ is a first stage action profile. Each SC entity's $i$ transfer amount (its first stage strategy) is given by the function $R^{i}:\Pi \rightarrow{A^{i}_{1}(\Pi)}$, where as $A^{i}_{1}$ we denote the set of all available first stage transfer actions of SC agency $i$. Following the budget transfer, the SC agency's $i$ defense endowment is given by: 
\begin{equation}
    d_{i}(r^{i\rightarrow{} j},r^{j\rightarrow{} i})=c_{i}-r^{i\rightarrow{} j}+r^{j\rightarrow{} i}~~\forall i,j\in\{1,2\},~i \neq j
\end{equation}

\looseness=-1 The SC entities' budget transfer is observed by the adversary T who, at the second stage of the game, decides on his resource allocation across the two battles/games (physical and cyber-social) and against the two SC defense opponents that perform no action in this stage. The action history after stage one is $h_{2}=\{r^{1\rightarrow{} 2},  r^{2\rightarrow{} 1}\}$, and $H_{2}$ is the set of all possible histories (SC alliance budget exchanges). Given $h_{2}$ the TO allocated budget $\tau_{1}$ to fight the physical SC battle and budget $\tau_{2}$ to fight at the cyber-social layer with $\tau_{1}+\tau_{2}\leq \tau$. Thus, the stage two action profile is $\{\tau_{1},\tau_{2}\}$ with $A^{T}_{2}(H_{2})$ being the set of all available  budget $\tau$ divisions across the two SC layers. The TO's strategy during this stage is expressed by the amount he chose to allocate to the physical attack effort and is given by a function $\mathcal{T}:H_{2} \rightarrow{A^{T}_{2}(H_{2})}$, i.e., $\tau_{1}=\mathcal{T}(h_{2})=\mathcal{T}(r^{1\rightarrow{} 2},  r^{2\rightarrow{} 1})$. It follows that $\tau_{2}=\tau-\tau_{1}$.

\looseness=-1 Entering the final stage of the SCDG the history is formed as $h_{3}=\{h_{2},\tau_{1},\tau_{2}\}=\{r^{1\rightarrow{} 2},  r^{2\rightarrow{} 1},\tau_{1},\tau_{2}\}$ with $H_{3}$ being the set of all possible histories up to this point. During this SCDG stage the adversary TO participates in parallel to two CBGs (physical and cyber-social) that model his interactions with the SC defenders across all targets. Thus for each front $i$ and against an SC entity $i$ we formulate two CBGs $\forall i \in \{1,2\}$, namely:
\begin{equation}\small
    CBG_{i}\{\{T,i\},\{\tau_{i},d_{i}\}, \{S^{T}_{i},S^{i}\}, \Theta_{i}, v_{i}, \{ U^{T}_{i},U^{i}\} \} 
\end{equation}
with budget allocation vectors across physical and cyber-social battlefields $k$ denoted as $\boldsymbol{t}_{i}=[\tau_{i,1},...,\tau_{i,k},...,\tau_{i,\theta_{i}}]$, and $\boldsymbol{d}_{i}=[d_{i,1},...,d_{i,k},...,d_{i,\theta_{i}}]$ for the TO T and SC entities $i$, $i \in \{1,2\}$ respectively with $\sum_{k=1}^{\theta_{i}} \tau_{i,k} \leq \tau_{i}$, $\tau_{i,k}\geq 0 $, and $\sum_{k=1}^{\theta_i} d_{i,k} \leq d_{i}$, $d_{i,k}\geq 0 $. As discussed in subsection \ref{sub:CBG} these two games yield mixed strategy NEs where the players' budget allocation vectors across battlefields as seen in subsection \ref{sub:CBG} consist of random variables $\tau_{i,k},d_{i,k}$ characterized by the univariate distribution functions $\{\mathfrak{T}_{i,k}\}_{k=1}^{\theta_{i}}$ and $\{\mathfrak{D}_{i,k}\}_{k=1}^{\theta_{i}}$, $\forall i \in \{1,2\}$ respectively for each SC target $k \in \Theta_{i}$.

The mixed strategies (the $\theta_{i}$-variate distribution functions as defined in subsection \ref{sub:CBG}) that express the distribution of budget for each player across the two CBGs' battlefields are: 
\begin{equation*}\small
    \begin{split}
        O^{1}(h_{3})= O^{1}(r^{1\rightarrow{} 2},  r^{2\rightarrow{} 1},\tau_{1},\tau_{2}),~
        O^{2}(h_{3})= O^{2}(r^{1\rightarrow{} 2},  r^{2\rightarrow{} 1},\tau_{1},\tau_{2})\\
        O^{T}_{1}(h_{3})= O^{T}_{1}(r^{1\rightarrow{} 2},  r^{2\rightarrow{} 1},\tau_{1},\tau_{2}),~
        O^{T}_{2}(h_{3})= O^{T}_{2}(r^{1\rightarrow{} 2},  r^{2\rightarrow{} 1},\tau_{1},\tau_{2})\\  
    \end{split}
\end{equation*}

\looseness=-1 The MSNEs characterize a state for the two games where the two SC defenders have chosen their optimal randomization over their budget allocation across battlefields (SC area targets for game 1, cyber-social spaces for game 2) and thus they cannot improve the SC protection by making a different choice. In addition, the MSNEs for the TO across the two CBGs he participates in, are two probability distributions that capture his $\tau_{1},\tau_{2}$ budget allocations over battlefields towards maximizing his expected utility, namely the number of physical areas and social-media environments he will successfully strike. For the proposed SCDG the use of mixed strategies for both fronts/games is motivated by the fact that both the TO and the SC entities have to randomize over their strategies towards preventing their opponent from guessing their potential action. 

Let us now define the overall strategy profile of each SCDG player, which is a collection of maps from all possible histories into available actions, namely:
\begin{equation}\small
    \begin{split}
&    \zeta^{i} \triangleq \{R^{i},O^{i}\}~~ \forall i \in \{1,2\}\\
&    \zeta^{T} \triangleq \{\mathcal{T},O^{T}_{1},O^{T}_{2}\}
    \end{split}
\end{equation}
where $\zeta^{i}$ are the strategies (collection of functions) of the city entity $i$, $i \in \{1,2\}$, and $\zeta^{T}$ denotes the TO's strategies. Thus, the strategy profile is $\zeta=\{\zeta^{1}, \zeta^{2}, \zeta^{T}\}$ and the set that contains all possible player strategies is denoted as $Z \triangleq \{Z^{1},Z^{2},Z^{T}\}$, where $Z^{p},~p\in\{1,2,T\}$ is the set containing all possible strategies of SCDG player $p$.

\looseness=-1 Given the allocation of budget of the three players to each battlefield of the two parallel CBGs (final stage), we will further define the SCDG's terminal history as $h_{terminal}=\{r^{1\rightarrow{} 2},  r^{2\rightarrow{} 1},\tau_{1},\tau_{2}, \boldsymbol{t}_{1}, \boldsymbol{t}_{2}, \boldsymbol{d}_{1}, \boldsymbol{d}_{2}\}$ and as $H_{terminal}$ the set of all possible terminal histories. Finally, as $\mathcal{H}=H_{1}\cup H_{2}\cup H_{3} \cup H_{terminal}$ we denote the set of possible histories.
Given the mixed strategies of each player, and the CBG definition in subsection \ref{sub:CBG} the SCDG payoff functions following the final stage are $\Psi^{i}:H_{terminal} \rightarrow{\mathbb{R}}, \forall i \in \{1,2\}$, and $\Psi^{T}:H_{terminal} \rightarrow{\mathbb{R}}$.
Since the strategy profile $\zeta^{p}$ of each player $p$, $p \in \{1,2,T\}$ determines the SCDG's action path (i.e., the $h_{terminal}$) we can express the payoffs as:
\begin{equation}\label{Eq7}\small
    \begin{split}
&    \Psi^{i}(\zeta^{1}, \zeta^{2}, \zeta^{T}) \triangleq \mathbb{E}\bigg[\sum_{k=1}^{\theta_{i}}u^{i}_{k}({\tau}_{i,k},{d}_{i,k})\bigg]\triangleq \mathbb{E}\bigg[U^{i}\bigg]\\
&    \Psi^{T}(\zeta^{1}, \zeta^{2}, \zeta^{T}) \triangleq \mathbb{E}\bigg[\sum_{i=1}^{2}\sum_{k=1}^{\theta_{i}}u^{T}_{i,k}({\tau}_{i,k},{d}_{i,k})\bigg]\triangleq \mathbb{E}\bigg[ U^T_1  + U^T_2\bigg]
    \end{split}
\end{equation}
where:
\begin{equation*}\small
    u^{i}_{k}({\tau}_{i,k},{d}_{i,k}) = \begin{cases}
                                        v_{i} &\text{if ${d}_{i,k}>{\tau}_{i,k}$}\\
                                        0 &\text{if ${d}_{i,k}<{\tau}_{i,k}$}\\
                                        \frac{v_{i}}{2} &\text{if ${d}_{i,k}={\tau}_{i,k}$}
\end{cases}~~~~ \forall i \in \{1,2\}
\end{equation*}
\begin{equation*}\small
u^{T}_{i,k}({\tau}_{i,k},{d}_{i,k})= v_{i} - u^{i}_{k}({\tau}_{i,k},{d}_{i,k})
\end{equation*}
with $d_{i,k}, \tau_{i,k}$ being the random variables that denote the budget allocated by the players to a battlefield $k$.
The formal definition of the complete information SCDG is: 
\begin{equation*}\small
\begin{split}
& SCDG\bigg\{ \{1,2,T\}, \{\mathcal{H}\}, \{Z\},\\
& \{R^{1}, R^{2}, \mathcal{T}, O^{1}, O^{2}, O^{T}_{1}, O^{T}_{2}\}, \{\Psi^{1}, \Psi^{2}, \Psi^{T}\} \bigg\}.
\end{split}
\end{equation*}

\begin{definition}\label{definition3}
A behavior strategy profile $\zeta^{*}\triangleq \{\zeta^{1*}, \zeta^{2*}, \zeta^{T*}\}$ in the strategy set $Z \triangleq \{Z^{1},Z^{2},Z^{T}\}$ is a Nash equilibrium of the SCDG with set of players $P \triangleq \{1,2,T\}$ if
\begin{equation}\small
    \Psi^{p}(\zeta^{1*}, \zeta^{2*}, \zeta^{T*})\geq \Psi^{p}(\zeta^{p}, \zeta^{-p*}),~ \forall \zeta^{p} \in Z^{p},~ p \in P.
\end{equation}
\end{definition}

\section{Subgame Perfect Nash Equilibrium of the Smart City Defense Game}\label{sec:SPNE}

Since the SCDG is a multi-stage complete information game in extensive form we define a subgame perfect Nash equilibrium that requires the strategy of each player to be optimal after every stage history and not just at the beginning of the game \cite{fudenberg1991game}.
\begin{definition}\label{definition4}
Given a stage $\epsilon$ history $h_{\epsilon}$, $G(h_{\epsilon})$ is an SCDG's subgame happening after $h_{\epsilon}$ and $\zeta^{p}|h_{\epsilon}$ is the restriction of player's $p$, $p \in \{1,2,T\}$ strategies to histories in $G(h_{\epsilon})$. Then a behavior strategy profile $\zeta$ is a subgame perfect Nash equilibrium if for every $h_{\epsilon}$, the restriction $\zeta|h_{\epsilon}$ is a Nash equilibrium in $G(h_{\epsilon})$.
\end{definition}
For such multi-stage games with observed actions we can verify that a strategy profile $\zeta$ is subgame perfect by ensuring that no player $p$ can increase his utility by deviating from $\zeta$ in a single stage and reverting to $\zeta$ for the rest of the game. This is verified by using the one-stage deviation principle for games with a finite number of stages.
\begin{theorem}\label{theorem1}
The SCDG strategy profile $\zeta^{*}$ is a subgame perfect Nash equilibrium (SPNE) if and only if it satisfies the one-stage-deviation condition that for all players $p$, $p \in \{1,2,T\}$, stages $\epsilon$, and histories $h_{\epsilon}$: 
\begin{equation}\small
    \begin{split}
&    \Psi^{p}(\zeta^{p*}, \zeta^{-p*}|h_{\epsilon})\geq \Psi^{p}(\zeta^{p}, \zeta^{-p*}|h_{\epsilon})\\
&    s.t.~~ \zeta^{p}(h_{\epsilon})\neq \zeta^{p*}(h_{\epsilon}) \\
&   \zeta^{p}_{|h_{\epsilon}}(h_{\epsilon + \omega})=\zeta^{p*}_{|h_{\epsilon}}(h_{\epsilon + \omega})\\
& \forall \omega > 0,~  \forall \zeta^{p} \in Z^{p},~\forall p \in \{1,2,T\}.
    \end{split}
\end{equation}
\end{theorem}
The proof of Theorem \ref{theorem1} can be found in \cite{fudenberg1991game}. In order to derive the SPNE, for the SCDG we will apply backward induction since the game is of complete information with observed actions and exactly three stages (a finite number) \cite{fudenberg1991game}. The process identifies the equilibria in the latest stages and moves up until the initial stage of the extensive form game. In our case the backward induction algorithm initially considers the payoffs obtained by the optimal choice of the three players in the final Colonel Blotto games stage (Nash Equilibrium) that maximizes their payoff. In what follows, we describe the backward induction process towards determining their SCDG SPNE, focusing on the budget allocation strategies of the TO, ESA, and ICT agency.

\subsection{Colonel Blotto Nash Equilibrium Payoffs}

First, we focus our attention on the payoffs of the three players at the Nash Equilibrium of the two CBGs that take place at the physical ($CBG_{1}$) and the cyber-social plane ($CBG_{2}$) of the smart city. Given the definition in subsection \ref{sub:CBG} and the analysis of the static CBG in \cite{roberson2006colonel} the payoffs for each player depend on the initial budgets $\tau_{i}$ (for the TO), $d_{i}$  (for each SC entity), and are given as follows.

For each static $CBG_{i}$ of value $\phi_{i}=|\Theta_{i}|\cdot v_{i}$ that takes place at the third stage of the SCDG there exists a Nash equilibrium with unique  payoff for an SC entity player $i$, $i \in \{1,2\}$ playing against the TO $T$. Each SC player's payoff is~\cite{hajimirsaadeghi2017dynamic}: 
\begin{equation}\label{Eq10}\small
U^{i}(\tau_{i},d_{i}) = \begin{cases}
            0,~~~~~~~~~~~~~~~~~~~~~~~~~~~ if~~\frac{d_{i}}{\tau_{i}}<\frac{1}{|\Theta_{i}|}\\
            
            \phi_{i}\Big(\frac{2 \cdot \beta - 2}{\beta \cdot |\Theta_{i}|^2}\Big),~~~~~~~~~~~~~~~ if~~\frac{1}{|\Theta_{i}|} \leq \frac{d_{i}}{\tau_{i}}<\frac{1}{|\Theta_{i}|-1}\\
            
            \phi_{i}\Big(\frac{2}{|\Theta_{i}|}-\frac{2\cdot \tau_{i}}{|\Theta_{i}|^2 \cdot d_{i}}\Big),~~~~~~ if~~ \frac{1}{|\Theta_{i}|-1} \leq \frac{d_{i}}{\tau_{i}} < \frac{2}{|\Theta_{i}|}\\
            
            \phi_{i}\cdot \frac{d_{i}}{2\cdot \tau_{i}},~~~~~~~~~~~~~~~~~~~  if~~ \frac{2}{|\Theta_{i}|} \leq \frac{d_{i}}{\tau_{i}} < 1\\
            
            \phi_{i}-\phi_{i}\cdot \frac{\tau_{i}}{2\cdot d_{i}},~~~~~~~~~~~~~ if~~ 1 \leq \frac{d_{i}}{\tau_{i}} < \frac{|\Theta_{i}|}{2}\\
            
            \phi_{i}-\phi_{i}\Big(\frac{2}{|\Theta_{i}|}-\frac{2\cdot d_{i}}{|\Theta_{i}|^2 \cdot \tau_{i}}\Big),~if~~ \frac{|\Theta_{i}|}{2} \leq \frac{d_{i}}{\tau_{i}} \leq |\Theta_{i}|-1\\
            
             \phi_{i}-\phi_{i}\Big(\frac{2 \cdot \beta' - 2}{\beta' \cdot |\Theta_{i}|^2}\Big),~~~~~~~~~if~~ |\Theta_{i}|-1 < \frac{d_{i}}{\tau_{i}} \leq |\Theta_{i}|\\
             
             \phi_{i},~~~~~~~~~~~~~~~~~~~~~~~~~~ if~~ |\Theta_{i}|<\frac{d_{i}}{\tau_{i}}
\end{cases}
\end{equation} 
where $\beta=\ceil*{\frac{\frac{d_{i}}{\tau_{i}}}{1-(|\Theta_{i}|-1)\frac{d_{i}}{\tau_{i}}}}$, and $\beta'=\ceil*{\frac{\frac{\tau_{i}}{d_{i}}}{1-(|\Theta_{i}|-1)\frac{\tau_{i}}{d_{i}}}}$. 
For $\frac{2}{|\Theta_{i}|} \leq \frac{d_{i}}{\tau_{i}} \leq \frac{|\Theta_{i}|}{2}$ equilibrium ties have zero probability and the tie rule plays no role. For all other ratios Eq. \ref{Eq10} gives a tie to the player with the larger budget, as in \cite{roberson2006colonel}, instead of splitting it. Accordingly the payoff of the TO for the $CBG_{i}$, $i \in \{1,2\}$ is:
\begin{equation}\label{Eq11}
    U^{T}_{i}(\tau_{i},d_{i}) = \phi_{i} - U^{i}(\tau_{i},d_{i})
\end{equation}
where $\tau_{i}$ is the TO budget allocated for game $i$ and $d_{i}$ the SC entity's $i$, $i \in \{1,2\}$ budget entering the SCDG final stage.

\looseness=-1 The authors in the seminal work \cite{roberson2006colonel} provide a proof of the existence of the equilibrium in the CBG. Determining the MSNE for the CBG and thus the $\theta$-variate distributions is not trivial and is an active research area \cite{roberson2012non, paarporn2026blotto}. A number of approaches have been proposed including fictitious play \cite{hajimirsadeghi2016inter}, and geometric methods \cite{thomas2012n} while the latest research works rely on linear and dynamic programming approaches to solve the discrete version of the game \cite{behnezhad2017faster, vu2018efficient}. Since it is out of the scope of this work we will omit MSNE construction details. Evidently, the final payoffs of the SC entities critically depend on the budget levels after the resource transfer which is the phenomenon we try to model in this work.

\looseness=-1 Given the definition of the SCDG payoffs for the three players as presented in Eq. \ref{Eq7}, there are 64 unique forms of the SCDG payoff function $\Psi^{T}$ for the SC adversary TO $T$ (8 possible payoffs from $CBG_{1}$ and another 8 from $CBG_{2}$). This leads to a vast number of SPNE that complicate the tractability of our solution. Therefore, in order to simplify our analysis, we will assume that the number of battlefields for the two games is arbitrarily large, which is physically supported by the fact that the examined SC environment consists of a very large number of possible physical targets and even larger number of social environments in the cyber space. In this case, the number of unique TO payoffs $\Psi^{T}$ collapses to 4 and Eq. \ref{Eq10}-\ref{Eq11} can be rewritten as:
\begin{equation}\label{Eq12}\small
U^{i}(\tau_{i},d_{i}) = \begin{cases}
            \phi_{i}\cdot \frac{d_{i}}{2\cdot \tau_{i}},~~~~~~~~~~  if~~ \frac{2}{|\Theta_{i}|} \leq \frac{d_{i}}{\tau_{i}} < 1\\
            \phi_{i}-\phi_{i}\cdot \frac{\tau_{i}}{2\cdot d_{i}},~~~~~ if~~ 1 \leq \frac{d_{i}}{\tau_{i}} < \frac{|\Theta_{i}|}{2}
\end{cases}
\end{equation} 
\begin{equation*}
U^{T}_{i}(\tau_{i},d_{i}) = \phi_{i} - U^{i}(\tau_{i},d_{i})
\end{equation*} 
Eq. \ref{Eq12} is the limit of Eq. \ref{Eq10} as $|\Theta_{i}|$ grows with $\phi_{i}$ fixed. Its two cases then cover every $d_{i}>0$ and $\tau_{i}\geq0$, with $U^{i}=\phi_{i}$ for $\tau_{i}=0$.

\subsection{Smart City Defense Game  Families of Equilibria}

Given the NE payoffs in the game's third stage, we compute the SPNE for the SCDG for each player during the second and first SCDG stage. In what follows, we define the net budget transfer from the ESA (SC entity 1) to the ICT agency (SC entity 2) as $r=r^{1\rightarrow{}2}-r^{2\rightarrow{}1}$.

\begin{theorem}\label{theorem3}
For the SCDG where  $\phi_1=|\Theta_1|\cdot v_1$ is the value of the physical CBG, $\phi_2=|\Theta_2|\cdot v_2$ is the value of the cyber-social CBG and prior to the third game stage:
\begin{itemize}
    \item the available budget of the ESA is $d_1=c_1-r$
    \item the available budget of the ICT agency is $d_2=c_2+r$
    \item the total available budget of the TO $T$ is $\tau$, and
    \item $\frac{2}{|\Theta_1|}<\frac{\tau}{d_1}<1$ and $\frac{2}{|\Theta_2|}<\frac{\tau}{d_2}<1$
\end{itemize}
then the second SCDG stage equilibrium strategy for the TO that maximizes its payoff is:
\begin{equation}\small
    \tau_{1}^{*} = \mathcal{T}^{*}(r^{1\rightarrow{}2},r^{2\rightarrow{}1})= \begin{cases}any~choice~\in [0,\tau],~~~ if~~ \frac{\phi_1}{d_1}=\frac{\phi_2}{d_2}\\
    \tau, ~~~~~~~~~~~~~ if~~ \frac{\phi_1}{d_1}>\frac{\phi_2}{d_2}\\
    0, ~~~~~~~~~~~~~ if~~ \frac{\phi_1}{d_1}<\frac{\phi_2}{d_2}
    \end{cases}
\end{equation}
\begin{equation*}\small
            \tau^{*}_2=\tau - \tau_{1}^{*}
    \end{equation*}
In this case if the SCDG parameters also satisfy either

\noindent
$\frac{\phi_1}{c_1}<\frac{\phi_2}{c_2}$, $\frac{2}{|\Theta_1|}<\frac{\tau}{c_1-\frac{\phi_2 c_1-\phi_1 c_2}{\phi_1+\phi_2}}<1$ \& $\frac{2}{|\Theta_2|}<\frac{\tau}{c_2+\frac{\phi_2 c_1-\phi_1c_2}{\phi_1+\phi_2}}<1$ 

\noindent
or

\noindent
$\frac{\phi_1}{c_1}>\frac{\phi_2}{c_2},\frac{2}{|\Theta_1|}<\frac{\tau}{c_1+\frac{\phi_1 c_2-\phi_2 c_1}{\phi_1+\phi_2}}<1$ \& $\frac{2}{|\Theta_2|}<\frac{\tau}{c_2-\frac{\phi_1 c_2-\phi_2 c_1}{\phi_1+\phi_2}}<1$, then the first SCDG stage equilibrium strategies for the two SC entities that maximize their payoff are:
\begin{equation}\small
    R^{*1}=r^{*1\rightarrow{}2}= \begin{cases}
    \rho \in [0,\frac{\phi_2 c_1-\phi_1 c_2}{\phi_1+\phi_2}),~~~if~~ \frac{\phi_1}{c_1}<\frac{\phi_2}{c_2}\\
    0,~~~~~~~~~~~~~~~~~~~~~ otherwise
    \end{cases}
\end{equation}
\begin{equation}\small
    R^{*2}=r^{*2\rightarrow{}1}= \begin{cases}
    \rho \in [0,\frac{\phi_1 c_2-\phi_2 c_1}{\phi_1+\phi_2}),~~~if~~ \frac{\phi_1}{c_1}>\frac{\phi_2}{c_2}\\
    0,~~~~~~~~~~~~~~~~~~~~~ otherwise
    \end{cases}
\end{equation}
where $\rho$ is such that $\tau<d_1$ and $\tau<d_2$ hold after the transfer, and there exists an SCDG SPNE family of $\{\mathcal{T}^{*},R^{*1},R^{*2}\}$ as defined above.
\end{theorem}
\begin{proof}
see Appendix \ref{appndx:proof3}.
\end{proof}

\looseness=-1 Theorem \ref{theorem3} characterizes the SPNE  actions of all SCDG participants in the case where the TO has the smallest available budget among all conflicting parties, following the SC budget transfer. In such a case the TO chooses to allocate his entire initial budget to a single CBG taking into account this game's value along with the strength of his opponent budget-wise. If both players are equally unattractive for the TO he randomizes his budget allocation towards the two fights-CBGs. In response, according to the SPNE, the SC entity whose plane is not under threat can transfer budget to the other SC player within a set as his payoff will not be affected by this action. The transfer can take place even if the party that provides resources has less initial budget than his ally. The existence of the upper bound in this transfer guarantees the TO's action and essentially the SPNE's existence.

\begin{theorem}\label{theorem4}
For the SCDG where  $\phi_1=|\Theta_1|\cdot v_1$ is the value of the physical CBG, $\phi_2=|\Theta_2|\cdot v_2$ is the value of the cyber-social CBG and prior to the third game stage:
\begin{itemize}
    \item the available budget of the ESA is $d_1=c_1-r$
    \item the available budget of the ICT agency is $d_2=c_2+r$
    \item the total available budget of the TO $T$ is $\tau$, and
    \item $d_1+d_2<\tau$, $\frac{2}{|\Theta_1|}<\frac{d_1}{\frac{\tau}{1+\sigma}}<1$ and $\frac{2}{|\Theta_2|}<\frac{d_2}{\frac{\sigma \cdot \tau}{1+\sigma}}<1$
\end{itemize}
where $\sigma=\sqrt{\frac{\phi_2 d_2}{\phi_1 d_1}}$,
then the second SCDG stage equilibrium strategy for the TO that maximizes its payoff is:
\begin{equation}
\begin{split}
    \tau_{1}^{*} = \mathcal{T}^{*}(r^{1\rightarrow{}2},r^{2\rightarrow{}1})=\frac{\tau}{1+\sqrt{\frac{\phi_2 d_2}{\phi_1 d_1}}}\\
        \tau^{*}_2=\tau - \tau_{1}^{*}
\end{split}
\end{equation}
In this case the first SCDG stage equilibrium strategies for the two SC entities that maximize their payoff are:
\begin{equation}
    R^{*1}=r^{*1\rightarrow{}2}= \begin{cases}
    \frac{c_1-c_2}{2}-\frac{c_1+c_2}{2}\cdot \sqrt{\frac{\phi_1}{\phi_1+\phi_2}},~~~if~~\frac{c_1-c_2}{2\sqrt{c_1 c_2}}>\sqrt{\frac{\phi_1}{\phi_2}}\\
    0,~~~~~~~~~~~~~~~~~~~~~~~~~~~ otherwise
    \end{cases}
\end{equation}
\begin{equation*}
    R^{*2}=r^{*2\rightarrow{}1}= 0
\end{equation*}
and there exists an SCDG SPNE family of $\{\mathcal{T}^{*},R^{*1},R^{*2}\}$ as defined above. For the case where the budget strength of the two SC allies is interchanged the same SPNE family exists with the reverse budget transfers.
\end{theorem}
\begin{proof}
see Appendix \ref{appndx:proof4}.
\end{proof}

Theorem \ref{theorem4} characterizes the SPNE actions of all SCDG participants when the SC entities are at a disadvantage and their budgets are significantly smaller than the total budget of the TO. In this case, the TO allocates budget to both the physical and social games. In response the SC entity with the highest pre-allocated defense budget ($c_i, i\in\{1,2\}$) chooses to transfer budget to its SC ally.

\begin{theorem}\label{theorem5}
For the SCDG where  $\phi_1=|\Theta_1|\cdot v_1$ is the value of the physical CBG, $\phi_2=|\Theta_2|\cdot v_2$ is the value of the cyber-social CBG and prior to the third game stage:
\begin{itemize}
    \item the available budget of the ESA is $d_1=c_1-r$
    \item the available budget of the ICT agency is $d_2=c_2+r$
    \item the total available budget of the TO $T$ is $\tau$, and
    \item  $\frac{2}{|\Theta_1|}<\frac{\tau-\delta(r)}{d_1(r)}<1$ and $\frac{2}{|\Theta_2|}<\frac{d_2(r)}{\delta(r)}<1$
\end{itemize}
where $\delta=\sqrt{\frac{\phi_2d_1d_2}{\phi_1}}$, then the second SCDG stage equilibrium strategy for the TO that maximizes its payoff is:
\begin{equation}\small
\begin{split}
    \tau_{1}^{*} = \mathcal{T}^{*}(r^{1\rightarrow{}2},r^{2\rightarrow{}1})=\mathcal{T}^{*}(r) = \tau - \sqrt{\frac{\phi_2 d_1\cdot d_2}{\phi_1}}\\
        \tau^{*}_2=\tau - \tau_{1}^{*}
\end{split}
\end{equation}
In this case the first SCDG stage equilibrium strategies for the two SC entities that maximize their payoff are:
\begin{equation}
    R^{*1}=r^{*1\rightarrow{}2}= \begin{cases}
   \frac{\frac{\phi_2\cdot(c_1+c_2)^2}{4\phi_1\cdot \tau^2}\cdot c_1 -c_2}{1+\frac{\phi_2\cdot(c_1+c_2)^2}{4\phi_1\cdot \tau^2}},~~~ if~~\frac{c_1+c_2}{2\tau}>\sqrt{\frac{\phi_1 c_2}{\phi_2 c_1}}\\
    0,~~~~~~~~~~~~~~~~~~~~~~~~ otherwise
    \end{cases}
\end{equation}
\begin{equation*}
    R^{*2}=r^{*2\rightarrow{}1}= 0
\end{equation*}
and there exists an SCDG SPNE family of $\{\mathcal{T}^{*},R^{*1},R^{*2}\}$ as defined above. For the case where the budget strength relation of the two SC allies in comparison to the TO is interchanged the same SPNE family exists with the reverse budget transfers.
\end{theorem}
\begin{proof}
see Appendix \ref{appndx:proof5}.
\end{proof}

\looseness=-1 Theorem \ref{theorem5} characterizes the SPNE actions of all SCDG participants when one SC entity is at a disadvantage with fewer resources than its opponent TO, while its ally has a larger budget than the TO, considering the TO budget allocated to each game. In this case, the TO allocates budget to both physical and social games taking into account the strength of the two opponents along with the significance of each fight. In response, the SC entity with the superior pre-allocated defense budget ($c_i, i\in\{1,2\}$) in comparison to the TO can transfer budget to its SC ally. The existence and exact amount of the transfer should ensure the safety of this entity's plane and it takes place only if it leads to a higher expected utility for both SC players (higher number of expected wins).

The following theorem extends the previous results to all budget relations among the SCDG participants.

\begin{theorem}\label{theorem6}
For the SCDG where  $\phi_1=|\Theta_1|\cdot v_1$ is the value of the physical CBG, $\phi_2=|\Theta_2|\cdot v_2$ is the value of the cyber-social CBG, $|\Theta_1|$, $|\Theta_2|$ are arbitrarily large so that Eq. \ref{Eq12} holds for every budget ratio, and prior to the third game stage:
\begin{itemize}
    \item the available budget of the ESA is $d_1=c_1-r>0$
    \item the available budget of the ICT agency is $d_2=c_2+r>0$
    \item the total available budget of the TO $T$ is $\tau$
\end{itemize}
if $\frac{\phi_1}{d_1}\leq\frac{\phi_2}{d_2}$, then the second SCDG stage equilibrium strategy for the TO that maximizes its payoff is:
\begin{equation}\label{EqTO}\small
    (\tau_1^{*},\tau_2^{*})=\begin{cases}
    (0,~\tau),~~~~~~~~~~~ if~~\tau\leq\delta\\
    (\tau-\delta,~\delta),~~~~~~ if~~\delta<\tau\leq\delta+d_1\\
    \big(\frac{\tau}{1+\sigma},~\frac{\sigma\cdot\tau}{1+\sigma}\big),~~ if~~\tau>\delta+d_1
    \end{cases}
\end{equation}
where $\delta=\sqrt{\frac{\phi_2 d_1 d_2}{\phi_1}}$, $\sigma=\sqrt{\frac{\phi_2 d_2}{\phi_1 d_1}}$. If $\frac{\phi_1}{d_1}=\frac{\phi_2}{d_2}$ and $\tau<d_1+d_2$, every choice with $\tau_1\leq d_1$, $\tau_2\leq d_2$ is also optimal. If $\frac{\phi_1}{d_1}>\frac{\phi_2}{d_2}$, the same holds with the indices interchanged. Given this TO strategy, if $\frac{\phi_1}{c_1}<\frac{\phi_2}{c_2}$, then the first SCDG stage equilibrium strategies for the two SC entities are $R^{*2}=r^{*2\rightarrow{}1}=0$ and $R^{*1}=r^{*1\rightarrow{}2}=r^{*}$ with:
\begin{enumerate}[(i)]
\item if $\tau\geq\frac{C}{2}\Big(1+\sqrt{\frac{\phi_1+\phi_2}{\phi_1}}\Big)$:\\
$r^{*}=\max\Big\{0,\frac{c_1-c_2}{2}-\frac{C}{2}\sqrt{\frac{\phi_1}{\phi_1+\phi_2}}\Big\}$
\item if $\frac{C}{2}\max\Big\{1,\sqrt{\frac{\phi_2}{\phi_1}}\Big\}<\tau<\frac{C}{2}\Big(1+\sqrt{\frac{\phi_1+\phi_2}{\phi_1}}\Big)$:\\
$r^{*}=\max\Big\{0,\frac{\xi\cdot c_1-c_2}{1+\xi}\Big\}$
\item if $\phi_2\leq\phi_1$ and $\frac{\phi_2C}{\phi_1+\phi_2}\leq\tau\leq\frac{C}{2}$: $r^{*}=\bar{r}$, with the TO choosing $\tau_2^{*}=\min\{\tau,d_2\}$
\item otherwise: any $r^{*}\in[0,\bar{r}]$ that satisfies $\phi_2(c_1-r^{*})(c_2+r^{*})\geq\phi_1\tau^2$ (with $\tau_2^{*}=\tau$ if $r^{*}=\bar{r}$), and $r^{*}=0$ if no such $r^{*}$ exists
\end{enumerate}
where $C=c_1+c_2$, $\xi=\frac{\phi_2C^2}{4\phi_1\tau^2}$ and $\bar{r}=\frac{\phi_2c_1-\phi_1c_2}{\phi_1+\phi_2}$. For the case where $\frac{\phi_1}{c_1}>\frac{\phi_2}{c_2}$ the same holds with the players' position interchanged, while if $\frac{\phi_1}{c_1}=\frac{\phi_2}{c_2}$ no budget exchange is performed.
\end{theorem}
\begin{proof}
see Appendix \ref{appndx:proof6}.
\end{proof}

\looseness=-1 Cases (i) and (ii) of Theorem \ref{theorem6} yield the transfers of Theorems \ref{theorem4} and \ref{theorem5}, respectively, while case (iv) includes Theorem \ref{theorem3}. In all cases the donor is the SC entity with the minimum $\frac{\phi_i}{c_i}$ value, and no SC entity is worse off in comparison to the case where no transfer takes place. While \cite{kovenock2012coalitional} identifies when a transfer between two allies is beneficial, Theorem \ref{theorem6} provides the equilibrium transfers themselves, in closed form and for all budget relations and game values. This includes the boundary where both games are equally attractive to the TO (case (iii)), and all transfers that keep the donor out of the attack (case (iv)).

In addition, given the payoffs of Eq. \ref{Eq12}, Theorem \ref{theorem6} has been formally verified by the Lean proof assistant \cite{demoura2021lean}, i.e., it has been proved in a form where a computer checks every step. The verified statements are that the TO's budget allocation is optimal for any transfers, that an SPNE exists for all budgets and game values, that the SPNE transfers are exactly those of cases (i)-(iv), and that no SC entity is worse off.

\begin{figure}[t]
\centering
	\includegraphics[width=\columnwidth]{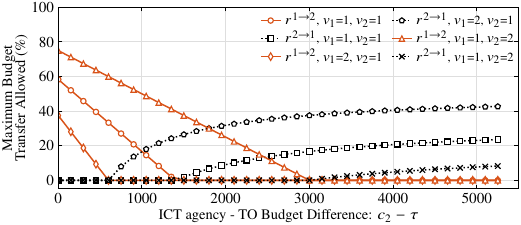}
\vspace{-0.05cm }
	\includegraphics[width=\columnwidth]{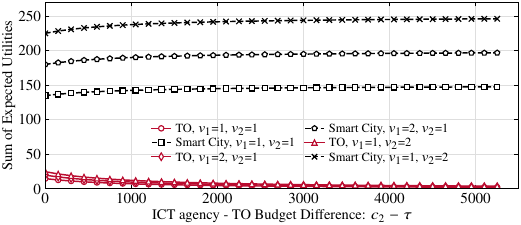} 
\caption{\small TO in budget disadvantage: (a, top) Budget transfer among agencies and (b, bottom) Expected Utility vs ICT agency - TO budget difference} \label{fig:2}
\vspace{-0.5cm}
\end{figure}

\section{Numerical Evaluation and Discussion}
In this section, we present a numerical evaluation of the SCDG focusing on the players' actions and response curves for various game parameters and budget strength relations among all parties. The SPNE transfers and the TO's budget allocation are given by Theorem \ref{theorem6}, while the expected utilities are computed by Eq. \ref{Eq10} for the given number of targets. A grid search with Eq. \ref{Eq10} showed that no player gains by deviating from the presented SPNE.\footnote{The Python code that solves the SCDG and reproduces the figures, along with the numerical and formal (Z3, Lean) checks of the theorems, is available at \url{https://github.com/dimisik/smart-city-defense-game}} In addition, we will present how the deviation from equilibrium strategies affects the payoffs of the defensive SC players towards compromising the public safety in the examined smart city setting.

\subsection{SCDG Analysis}

\looseness=-1 First, we focus on the case where both SC entities have greater defense budgets than their adversary TO in an SC setting that consists of $|\Theta_1|=50$ physical battlefields and $|\Theta_2|=100$ cyber-social battlefields. We consider a TO whose initial budget is $\tau=200$, an SC ESA with initial budget $c_1=800$, and evaluate how the maximum allowed transfer amount between the SC entities changes as the budget difference between the ICT agency and the TO increases. Fig. \ref{fig:2}a shows these results for different game values, namely when a) the physical battles are more important for the two opponents ($v_1>v_2$), b) the social battles are critical for the two opponents ($v_1<v_2$), or both planes are equally important ($v_1=v_2$). We observe that for small budget differences the ESA makes a transfer to the ICT agency up to a certain point that depends on the value of each game/plane. After this point, the TO chooses to allocate all his budget to fight the physical battle therefore a budget transfer from the ICT to the ESA is now optimal for the SC defense. We also observe that the SC entity transferring resources is not always the one with the highest budget but the one with the minimum $\frac{\phi_i}{c_i},~i\in\{1,2\}$ value, as sometimes the resourceful agency may have SC battlefields of higher importance to fight for.

\begin{figure}[t]
\begin{subfigure}{0.49\linewidth}
\includegraphics[width=\linewidth,keepaspectratio=true]{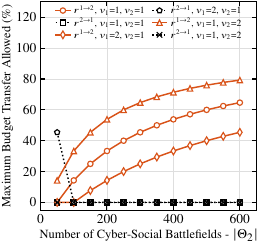}\end{subfigure}
\hspace*{\fill} 
\begin{subfigure}{0.49\linewidth}
\includegraphics[width=\linewidth,keepaspectratio=true]{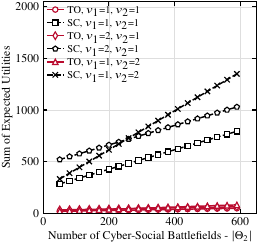}\end{subfigure}
\hspace*{\fill} 
\caption{\small TO in budget disadvantage: (a, left) Budget transfer among agencies, and (b, right) Expected Utility vs Social battlefields Number} \label{fig:3}
\vspace{-0.5cm}
\end{figure}

\begin{figure*}[t]
  \begin{subfigure}{0.32\textwidth}
     \includegraphics[width=\textwidth,keepaspectratio=true]{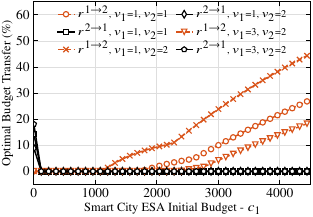}
    \caption{\footnotesize SC Agencies' Optimal Budget Transfer}
  \label{fig:4a}
 \end{subfigure}
  \begin{subfigure}{0.32\textwidth}
     \includegraphics[width=\textwidth,keepaspectratio=true]{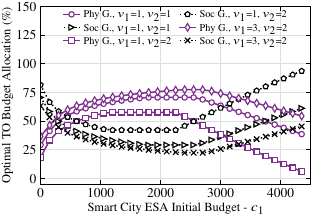}
      \caption{\footnotesize TO's Optimal Budget Allocation across Games}
     \label{fig:4b}
  \end{subfigure}
  \begin{subfigure}{0.32\textwidth}
     \includegraphics[width=\textwidth,keepaspectratio=true]{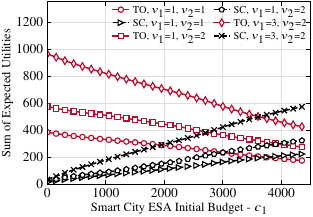}
      \caption{\footnotesize Players' Expected Utilities}
     \label{fig:4c}
  \end{subfigure}
   \vspace{-0.1cm }
  \caption{\small Smart City Defense Game Strategies vs  ESA Initial Budget
  }
  \label{fig:4}
    \vspace{-0.5cm }
 \end{figure*}

In Fig. \ref{fig:2}b for the same set of parameters we observe how the expected utility of the TO and the expected utility of the SC as a whole (both agencies) changes as the budget of the ICT agency increases. Note that in the context of the two CBGs the expected utility is analogous to the total number of physical and cyber-social battlefields that were won by each player (successfully attacked by the TO, or successfully defended by the SC agencies).

\begin{figure*}[t]
  \begin{subfigure}[t]{0.32\textwidth}
     \includegraphics[width=\textwidth,keepaspectratio=true]{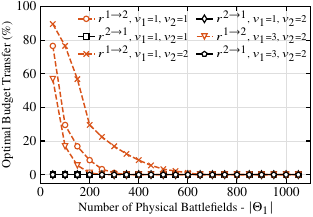}
    \caption{\footnotesize SC Agencies' Optimal Budget Transfer}
  \label{fig:5a}
 \end{subfigure}
  \begin{subfigure}[t]{0.32\textwidth}
     \includegraphics[width=\textwidth,keepaspectratio=true]{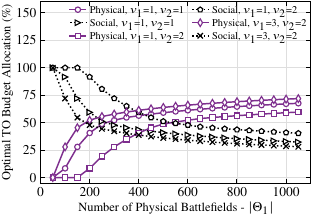}
      \caption{\footnotesize TO's Optimal Budget Allocation across Games}
     \label{fig:5b}
  \end{subfigure}
  \begin{subfigure}[t]{0.32\textwidth}
     \includegraphics[width=\textwidth,keepaspectratio=true]{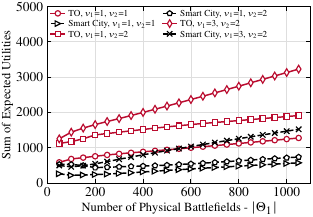}
      \caption{\footnotesize Players' Expected Utilities}
     \label{fig:5c}
  \end{subfigure}
    \vspace{-0.1cm }
  \caption{\small Smart City Defense Game Strategies vs  Number of Physical Battlefields ($\tau=1000,~c_1=500,~c_2=150$)}
  \label{fig:5}
  \vspace{-0.5cm }
 \end{figure*}

\looseness=-1 In Fig. \ref{fig:3} we evaluate how the SC agencies' budget transfer (Fig. \ref{fig:3}a), and expected utilities (Fig. \ref{fig:3}b) change as the number of cyber-social battlefields increases. The number of physical battlefields is $|\Theta_1|=250$, while the initial budgets for the TO, ESA and ICT agency are $\tau=200,~c_1=1500,~c_2=300$, respectively. Evidently, the budget of the ESA prohibits the TO from allocating any resources to the physical fight, therefore we observe transfers from the ESA to the ICT agency whose amount increases as the number of social fights increases. The only exception is the case of $v_1>v_2$ with a small number of social battlefields, where the TO attacks the physical plane and the ICT agency can transfer budget to the ESA. Those transfers lead to a higher expected utility sum for the SC as seen in Fig. \ref{fig:3}b in comparison to the TO whose limited resources reduce the probability of landing successful attacks.

\begin{figure*}[t]
  \begin{subfigure}[t]{0.32\textwidth}
     \includegraphics[width=\textwidth,keepaspectratio=true]{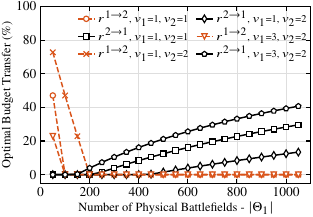}
    \caption{\footnotesize SC Agencies' Optimal Budget Transfer}
  \label{fig:6a}
 \end{subfigure}
  \begin{subfigure}[t]{0.32\textwidth}
     \includegraphics[width=\textwidth,keepaspectratio=true]{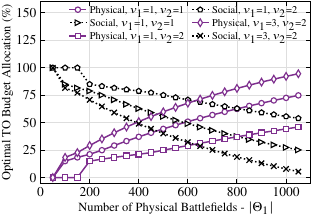}
      \caption{\footnotesize TO's Optimal Budget Allocation across Games}
     \label{fig:6b}
  \end{subfigure}
  \begin{subfigure}[t]{0.32\textwidth}
     \includegraphics[width=\textwidth]{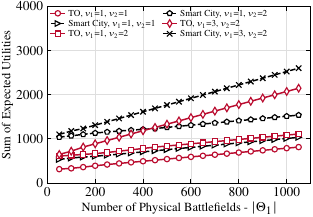}
      \caption{\footnotesize Players' Expected Utilities}
     \label{fig:6c}
  \end{subfigure}
    \vspace{-0.1cm }
  \caption{\small Smart City Defense Game Strategies vs  Number of Physical Battlefields ($\tau=1000,~c_1=150,~c_2=1200$)}
  \label{fig:6}
  \vspace{-0.6cm }
 \end{figure*}

Next, we examine the case where the TO has an advantage in comparison to at least one of the SC agencies.
\begin{figure}[t]
\begin{subfigure}{0.49\linewidth}
\includegraphics[width=\linewidth,keepaspectratio=true]{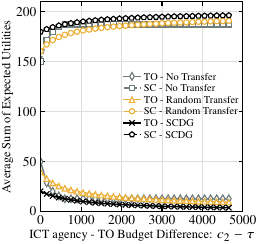} \end{subfigure}
\hspace*{\fill} 
\begin{subfigure}{0.49\linewidth}
\includegraphics[width=\linewidth,keepaspectratio=true]{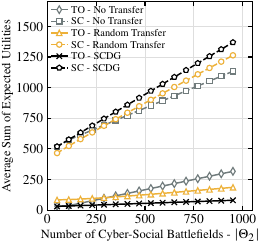}\end{subfigure}
\hspace*{\fill} 
\caption{\small Average Sum of Expected Utilities vs (a, left) ICT agency - TO Budget Difference, (b, right) Number of Cyber-Social Battlefields  } \label{fig:7}
\vspace{-0.3cm }
\end{figure}

\looseness=-1 In Fig. \ref{fig:4} we consider an SC setting of $|\Theta_1|=200$ physical and $|\Theta_2|=200$ cyber-social targets and a TO with available budget of $\tau=3500$ which is always larger than the initial budget of the ICT agency which is $c_2=300$. In this case, we examine how increasing the initial budget $c_1$ of the ESA affects the SPNE strategies for various target values $v_1,v_2$. Fig. \ref{fig:4a} shows how the optimal budget transfer among organizations changes. Initially, the ICT agency transfers budget to the ESA, while as the latter becomes more resourceful the opposite transfer takes place. Accordingly, Fig. \ref{fig:4b} shows how the TO responds to the budget transfer between the two agencies. Initially, since the ICT agency has more budget available, the majority of the TO's budget is allocated to the social battlefields. As the budget of the ESA increases the TO allocates less budget to attack the social targets in an attempt to counter the stronger opponent at the physical plane. This behavior that maximizes the expected number of wins in the two types of targets, stops after a critical point where the ESA initial budget is significantly greater than the ICT agency's budget. After this point (see Fig. \ref{fig:4b}) the TO starts allocating more budget to the social game as now this is where the SC is vulnerable. Finally, Fig. \ref{fig:4c} shows the sum of expected utility for the TO and the SC in general. As the total SC defense budget increases so does its expected utility, namely the number of social and physical spaces that will be successfully defended.

\looseness=-1 Next, in Figures \ref{fig:5},\ref{fig:6}, we evaluate how the SCDG SPNE strategies change as the TO decides to perform attacks against an increasing number of physical targets, while the social spaces under attack remain constant with $|\Theta_2|=800$. When the optimal transfer is not unique (case (iv) of Theorem \ref{theorem6}) the maximum allowed transfer is presented. Two different cases are considered. In Fig. \ref{fig:5} the TO has greater initial budget than his two opponents with $\tau=1000, c_1=500$, and $c_2=150$. When the number of physical targets under attack is small we observe a budget transfer (see Fig. \ref{fig:5a}) from the resourceful SC entity (here the ESA) to its ally. This transfer is, however, decreasing when the number of physical battlefields grows significantly. The TO's optimal budget allocation is seen in Fig. \ref{fig:5b}, while in Fig. \ref{fig:5c} we observe the expected number of won battles for the SCDG opponents. Evidently, the TO initially allocates more resources to the weaker opponent. As the number of his physical targets increases, it is forced to increase the budget allocation towards the physical CBG. Finally, as seen in Fig. \ref{fig:5c}, the initial budget advantage of the TO ($\tau>c_1,\tau>c_2$) is a significant factor as it always leads to a greater sum of expected utility and therefore a greater number of successful hits. In this case, the budget transfers between the two SC entities described by the SPNE strategies is the equilibrium response that will not increase the losses in either SC plane.

\begin{figure}[t]
\begin{subfigure}{0.49\linewidth}
\includegraphics[width=\linewidth,keepaspectratio=true]{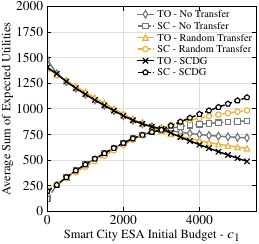}\end{subfigure}
\hspace*{\fill} 
\begin{subfigure}{0.49\linewidth}
\includegraphics[width=\linewidth,keepaspectratio=true]{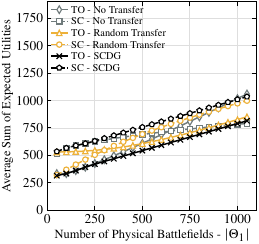}\end{subfigure}
\hspace*{\fill} 
\caption{\small Average Sum of Expected Utilities vs (a, left) ESA Initial Budget ($\tau=3500,~c_2=1000$), (b, right) Number of Physical Battlefields ($\tau=1000,~c_1=150,~c_2=1200$) } \label{fig:8}
\vspace{-0.6cm }
\end{figure}

\looseness=-1 On the contrary, in Fig. \ref{fig:6} we examine the case where the ICT agency has a greater budget than both the TO and the ESA. Again, budget transfer (see Fig. \ref{fig:6a}) occurs to reinforce the weaker SC player. Only when the number of physical targets is very small, the TO attacks the social plane alone and the ESA can transfer budget to the ICT agency. In addition, as the weakest ally has to defend an increasing number of physical targets, the optimal budget transfer percentage from his ally increases as well. In response, as seen in Fig. \ref{fig:6b} the TO allocates larger budget amounts to the physical fight as the number of the physical targets under attack increases. Finally, Fig. \ref{fig:6c} shows the sum of expected utilities for the TO and the SC (combined utility of the two agencies) for this specific parameter set as the number of physical targets increases.

\begin{figure}[t]
\centering
\includegraphics[width=0.85\columnwidth]{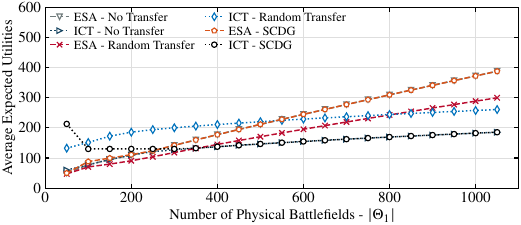}
\caption{\small Average Expected Utilities for SC Agencies vs Number of Physical Battlefields ($\tau=1000,~c_1=500,~c_2=150$)}
\label{fig:9}
\vspace{-0.6cm}
\end{figure}

\subsection{Comparative Analysis}

In this subsection we present comparative results that showcase how the deviation from the SPNE strategies for the two SC entities affect their utilities, and ability to defend physical and social targets, introducing vulnerabilities into the SC setting. In what follows we consider three budget transfer strategies, between SC allies, namely:
\begin{enumerate}[(a)]
    \item no budget transfer occurs
    \item a random transfer between the ESA and the ICT agency takes place
    \item both SC agencies act according to the SCDG SPNE strategies
\end{enumerate}
In all cases the TO responds optimally to the resulting SC budgets. For these cases, we present the average sum of expected utilities. While for the case (a), and (c) the results are analytical, for the random transfer we averaged the sum of expected utilities from $10^4$ simulations. Again we will examine different cases regarding the initial budget strength relations among SCDG players.

\looseness=-1 First, the case where the TO is in a budget disadvantage. In Fig. \ref{fig:7}a an SC setting with $|\Theta_1|=|\Theta_2|=100$ physical and social targets is considered with their values being equal $v_1=v_2=1$. The initial budget of the TO and ESA is $\tau=200$, and $c_1=800$ respectively. We present the average sum of expected utilities as the budget of the ICT agency increases in relation to the TO total budget. Evidently, when the SCDG strategies are followed, the utility sum of the SC is the highest (more successfully defended socio-physical targets) while the opposite happens for the TO. The same behavior is observed in Fig. \ref{fig:7}b where $|\Theta_1|=500$, $c_1=1500$, the ICT's budget is set to $c_2=300>\tau=200$, and the average sum of expected utilities is evaluated against an increasing number of social spaces targeted by the TO.

\looseness=-1 Next we consider the case where the TO has greater budget than at least one opponent in an SC setting where $|\Theta_2|=800$, $v_1=v_2=1$. Fig. \ref{fig:8}a shows the average sum of expected utilities as the ESA's initial budget increases, for the case where $|\Theta_1|=800$, $\tau=3500$ and $c_2=1000$. Fig. \ref{fig:8}b shows again the average utilities when the TO decides to target an increasing number of physical targets and the budgets of the TO, ESA, and ICT agency are $\tau=1000,~c_1=150,~c_2=1200$, respectively. Evidently, in both cases, the SCDG strategies for the two SC entities lead to at least as many successfully defended socio-physical targets as the case of no transfer. A random transfer can lead to a larger utility sum in some cases (e.g., in parts of Fig. \ref{fig:8}a, and for the parameters of Fig. \ref{fig:9}), but it is not beneficial for both SC entities. Finally, Fig. \ref{fig:9} shows the average expected utilities for each SC entity separately as the number of physical targets under attack increases and the players' budgets are $\tau=1000,~c_1=500,~c_2=150$. This figure showcases an important property of the SCDG. For the ICT agency, a random transfer yields a higher expected utility/number of wins, except when the number of physical targets is very small. However, as this number increases, the SPNE strategy forbids the two SC allies from making a budget transfer. This happens because the proposed game allows budget transfers only if neither ally is worse off. In our case when a random budget transfer is considered the average expected utility of the ESA is lower than the SPNE strategy of the SCDG, thus randomness is not beneficial for both allies and both SC planes.

\section{Conclusion}
In this work, we demonstrate a budget management mechanism between Smart City agencies deployed in cases of simultaneous terrorist attacks on multiple city levels and targets. The Smart City is modeled as a setting with two parallel layers, namely a physical, and a cyber-social. Each layer contains multiple targets/spaces, either physical (e.g., landmarks), or social (e.g., X (formerly Twitter) feeds) and their defense is assigned to two city agencies. A terrorist organization allocates budget to attack both SC layers and as a defense measure, the two agencies make budget transfers between them before allocating their resources among targets. In order to capture their interactions and define the optimal strategies that will improve the SC defenses, we propose the Smart City Defense Game (SCDG) which is a multi-stage extensive form game and derive its sub-game perfect Nash equilibrium. The proposed model provides strategies for budget exchanges between SC allies in cases of terrorist threats by considering the response and resource allocations of the enemy across the two SC planes. We show detailed numerical results for various parameter regions where when the SC agencies act according to the SPNE, they manage to increase the number of defended targets and reduce the cases where the terrorist organization launches successful attacks, in comparison to the case of no transfer.

\appendices
\section{Proof of Theorem \ref{theorem3}}\label{appndx:proof3}
The expected Nash equilibrium payoff functions of the three SCDG players after stage three are given by Eq. \ref{Eq12} depending on the ratio of available player budgets. The TO during the second stage reacts to the budget allocation of the SC entities ($r^{1\rightarrow{}2}$, $r^{2\rightarrow{}1}$) and allocates his budget $\tau$ in an effort to maximize his expected payoff. For our simplified case where $\frac{2}{|\Theta_1|}<\frac{\tau}{d_1}<1$ and $\frac{2}{|\Theta_2|}<\frac{\tau}{d_2}<1$ the expected payoff of $T$ as a function of his own budget allocation across the physical ($\tau_1$) and social ($\tau_2$) battles is:
\begin{equation*}\small
    \Psi^{T}(\tau_1)=\phi_1 \cdot \frac{\tau_1}{2d_1} + \phi_2 \cdot \frac{\tau_2}{2d_2} 
= \phi_1 \frac{\tau_1}{2d_1}+\phi_2 \frac{\tau-\tau_1}{2d_2}
\end{equation*}
The first derivative is $\frac{\partial \Psi^{T}}{\partial \tau_{1}}=\frac{\phi_1}{2d_1}-\frac{\phi_2}{2d_2}$ and we have to consider three distinct cases:
\begin{enumerate}[a)]
    \item $ \frac{\partial\Psi^{T}}{\partial \tau_1}=0 \Leftrightarrow \frac{\phi_1}{2d_1}=\frac{\phi_2}{2d_2}$, thus any budget allocation $\tau_1 \in [0,\tau]$ is optimal for the TO
    \item $ \frac{\partial\Psi^{T}}{\partial \tau_1}>0 \Leftrightarrow \frac{\phi_1}{2d_1}>\frac{\phi_2}{2d_2}$, then $\Psi^{T}$ is increasing in $\tau_1 \in [0,\tau]$ and will be maximum at $\Psi^{T}(\tau_1=\tau)= \frac{\phi_1 \tau}{2 d_1}$ which means that the TO will allocate all the budget fighting the physical SC game (CBG 1)
    \item $ \frac{\partial\Psi^{T}}{\partial \tau_1}<0 \Leftrightarrow \frac{\phi_1}{2d_1}<\frac{\phi_2}{2d_2}$, then following the logic of case b the TO allocates all his budget to the social game (CBG 2)
\end{enumerate}
Therefore:
\begin{equation*}\small
    \tau_{1}^{*} = \mathcal{T}^{*}(r^{1\rightarrow{}2},r^{2\rightarrow{}1})= \begin{cases}any~choice~\in [0,\tau],~~~ if~~ \frac{\phi_1}{d_1}=\frac{\phi_2}{d_2}\\
    \tau, ~~~~~~~~~~~~~~~~~~~~~~~~~ if~~ \frac{\phi_1}{d_1}>\frac{\phi_2}{d_2}\\
    0, ~~~~~~~~~~~~~~~~~~~~~~~~~ if~~ \frac{\phi_1}{d_1}<\frac{\phi_2}{d_2}
    \end{cases}
\end{equation*}
\begin{equation}
            \tau^{*}_2=\tau - \tau_{1}^{*}
    \end{equation}

\looseness=-1 We will now focus on the first SCDG stage, where the two SC entities should decide on the budget transfer between them. In this stage the known SCDG parameters are the two CBGs' values ($\phi_1,\phi_2$) and each SC entity's emergency response budget $c_1$,$c_2$. Assume $\frac{\phi_1}{c_1}<\frac{\phi_2}{c_2}$. In this case if no budget transfer is performed during stage 1 the TO will allocate all his budget to the social game 2 ($\tau^{*}_{1}=0,\tau^{*}_{2}=\tau$) according to the aforementioned second stage response (for $\tau\geq c_2$ see Appendix \ref{appndx:proof6}). A positive transfer from player two (ICT) to player one (ESA) will reduce the payoff of player two and will have no impact on the payoff of player one. Thus, $r^{2\rightarrow{}1}=0$ for $\frac{\phi_1}{c_1}<\frac{\phi_2}{c_2}$. Let us now assume that a positive transfer will occur from the ESA to the ICT agency while maintaining the conditions that will trigger the same TO response in stage two\footnote{This requirement stems from the one-stage deviation principle, Section \ref{sec:SPNE} - Theorem \ref{theorem1}}, namely $\frac{2}{|\Theta_1|}<\frac{\tau}{c_1-r^{1\rightarrow{}2}}<1$, $\frac{2}{|\Theta_2|}<\frac{\tau}{c_2+r^{1\rightarrow{}2}}<1$, and $\frac{\phi_1}{c_1-r^{1\rightarrow{}2}}<\frac{\phi_2}{c_2+r^{1\rightarrow{}2}}$.
Then we can calculate the maximum budget transfer that improves the payoff of player two and maintains the payoff of player one (Pareto improving transfer \cite{kovenock2012coalitional}) as:
\begin{equation}\small
    \frac{\phi_1}{c_1-r^{1\rightarrow{}2}} < \frac{\phi_2}{c_2+r^{1\rightarrow{}2}} \Leftrightarrow r^{1\rightarrow{}2} < \frac{\phi_2 c_1 - \phi_1 c_2}{\phi_1+\phi_2}
\end{equation}
The ESA (player 1) will never transfer budget that exceeds $\frac{\phi_2 c_1 - \phi_1 c_2}{\phi_1+\phi_2}$ since it would lead to a different TO response that would reduce his payoff. Since the TO assigns all his budget to fight the social CBG, the ESA (acting according to the SPNE) is allowed to transfer up to $\frac{\phi_2 c_1 - \phi_1 c_2}{\phi_1+\phi_2}$ in order to maintain this response and aid the ICT agency. The analysis is analogous for the $\frac{\phi_1}{c_1}>\frac{\phi_2}{c_2}$ case. Finally when $\frac{\phi_1}{c_1}=\frac{\phi_2}{c_2}$ no transfer guarantees an improvement for the SC entities' payoff thus no budget exchange is performed.

\section{Proof of Theorem \ref{theorem4}}\label{appndx:proof4}

For the specific parameters of this case ($d_1+d_2<\tau$, and $\frac{d_1}{\tau_1}<1$, $\frac{d_2}{\tau_2}<1$) the expected payoff of $T$ as a function of his own budget allocation across the physical ($\tau_1$) and social ($\tau_2$) battles following Eq. \ref{Eq12} is:
\begin{equation}\small
    \begin{split}
        \Psi^{T}(\tau_1)=\phi_1-\phi_1\frac{d_1}{2\tau_1}+\phi_2 - \phi_2 \cdot \frac{d_2}{2\tau_2}
        \stackrel{\tau_2=\tau-\tau_1}{=}\\
        \phi_1-\phi_1\frac{d_1}{2\tau_1}+\phi_2 - \phi_2 \cdot \frac{d_2}{2(\tau-\tau_1)}
    \end{split}
\end{equation}
The first derivative is:
\begin{equation}\small
        \frac{\partial\Psi^{T}(\tau_1)}{\partial \tau_1}=\frac{\phi_1 d_1}{2(\tau_1)^2}-\frac{\phi_2 d_2}{2(\tau - \tau_1)^2}
\end{equation}        
The budget allocation to $CBG_1$ that will maximize the TO's payoff is:         
 \begin{equation*}\small       
        \frac{\partial\Psi^{T}(\tau_1)}{\partial \tau_1}=0 \Leftrightarrow \\
        \frac{(\tau - \tau_1)^2}{(\tau_1)^2}=\frac{\phi_2 d_2}{\phi_1 d_1}\Leftrightarrow \frac{\tau - \tau_1}{\tau_1}=\sqrt{\frac{\phi_2 d_2}{\phi_1 d_1}}
\end{equation*}   

since $\tau-\tau_1$ is a strictly positive quantity. Thus, $\tau^{*}_1=\frac{\tau}{1+\sqrt{{\frac{\phi_2 d_2}{\phi_1 d_1}}}}$.
Also, since $\frac{\partial^2 \Psi^{T}(\tau_1)}{\partial \tau_1^2}= - \frac{\phi_1 d_1 \tau_1}{(\tau_1)^4} -\frac{\phi_2 d_2 (\tau - \tau_1)}{(\tau - \tau_1)^4}<0$ as $\tau-\tau_1>0$, $\Psi^{T}(\tau_1)$ is concave and $\tau^{*}_1$ a maximum. Given that, without loss of generality we assume that the ESA (player 1) transfers positive budget equal to $r,~r\geq0$ to the ICT agency (player 2), while their initial budget is $c_1, c_2$ respectively. Then the ESA payoff (Eq. \ref{Eq12}) is: 
\begin{equation}\small
    \Psi^{1}(r)=\phi_1 \frac{c_1-r}{2 \tau^{*}_1(r)} = \phi_1 \frac{c_1-r}{2 \cdot \frac{\tau}{1+\sqrt{\frac{\phi_2(c_2+r)}{\phi_1(c_1-r)}}}}
\end{equation}
Finding the first and second derivative yields:
\begin{equation*}\small
\frac{\partial\Psi^{1}(r)}{\partial r}= - \frac{\phi_1}{2\tau}+\frac{\sqrt{\phi_1 \phi_2}}{4\tau} \cdot \frac{c_1-c_2 -2r}{ \sqrt{c_1-r} \cdot \sqrt{c_2+r}}
\end{equation*}
and
\begin{equation*}\small
\frac{\partial^2\Psi^{1}(r)}{\partial r^2}=-\dfrac{\sqrt{\phi_1 \phi_2}\left(c_2+c_1\right)^2}{8\tau\left(c_1-r\right)^\frac{3}{2}\left(c_2+r\right)^\frac{3}{2}}
\end{equation*}
Since $\frac{\partial^2\Psi^{1}(r)}{\partial r^2}<0$, $ \Psi^{1}(r)$ is concave. Evidently, it is beneficial for the ESA to transfer to the ICT agency iff at the beginning of the domain of definition:
\begin{equation*}\small
\begin{split}
\frac{\partial\Psi^{1}(r)}{\partial r}\bigg|_{r=0}>0 \Leftrightarrow - \frac{\phi_1}{2\tau}+\frac{\sqrt{\phi_1 \phi_2}}{4\tau} \cdot \frac{c_1-c_2}{ \sqrt{c_1} \cdot \sqrt{c_2}}>0 \\
\Leftrightarrow \frac{c_1-c_2}{2\sqrt{ c_1 c_2}}>\sqrt{\frac{\phi_1}{\phi_2}}
\end{split}
\end{equation*}
This is a necessary condition for the existence of a budget transfer from player 1 to player 2 that is mutually beneficial. The condition also implies that $c_1-c_2>2\sqrt{\frac{\phi_1}{\phi_2}}\sqrt{c_1 c_2}>0 \Rightarrow c_1>c_2$. Regarding the optimal amount of budget to be transferred $r^{*1\rightarrow{}2}$, it is given by:
\begin{equation*}\small
\begin{split}
\frac{\partial\Psi^{1}(r)}{\partial r}=0 \Rightarrow ... \Rightarrow\\
r^2 - (c_1-c_2)r+\frac{1}{4}\frac{\phi_2(c_1-c_2)^2}{\phi_1+\phi_2}-\frac{\phi_1}{\phi_1+\phi_2}c_1c_2=0 \Leftrightarrow ... \Leftrightarrow\\
 r=\frac{c_1-c_2}{2}\pm\frac{c_1+c_2}{2}\sqrt{\frac{\phi_1}{\phi_1+\phi_2}}
\end{split}
\end{equation*}
Since $\frac{\partial\Psi^{1}(r)}{\partial r}=0$ requires $c_1-c_2-2r>0$, i.e., $r<\frac{c_1-c_2}{2}$, the root with the plus sign (introduced by squaring) is rejected. Thus, the only viable solution is $r^{*1\rightarrow{}2}=\frac{c_1-c_2}{2}-\frac{c_1+c_2}{2}\sqrt{\frac{\phi_1}{\phi_1+\phi_2}}$. Transfers that violate the conditions of this case are not beneficial for player 1 either (Appendix \ref{appndx:proof6}).

Regarding the opposite transfer ($r^{*2\rightarrow{}1}$) if we assume that a quantity $r$ is transferred from the ICT agency (player 2) to the ESA then the payoff of player 2 is given by:
\begin{equation}\small
\begin{split}
    \Psi^{2}(r)=\frac{\phi_2}{2} \frac{c_2-r}{\tau^{*}_2(r)} = \frac{\phi_2}{2} \frac{c_2-r}{\tau-\tau^{*}_1(r)} = ... = \\ \frac{\sqrt{\phi_1\phi_2}}{2\tau}(\sqrt{(c_2-r)(c_1+r)})+\frac{\phi_2}{2\tau}(c_2-r)
\end{split}
\end{equation}
The first derivative:
\begin{equation}
\begin{split}
\frac{\partial\Psi^{2}(r)}{\partial r}= ... = -(\frac{\sqrt{\phi_1\phi_2}}{4\tau}\cdot \frac{c_1-c_2+2r}{\sqrt{(c_2-r)(c_1+r)}}+\frac{\phi_2}{2\tau})<0
\end{split}
\end{equation}
since $c_1>c_2$ in our general case for $r \in [0,c_2)$.
Thus, $r^{*2\rightarrow{}1}=0$ always since no transfer is beneficial for SC player 2.

The same analysis can be followed for the case of $c_2>c_1$ and $\frac{c_2-c_1}{2\sqrt{c_1c_2}}>\sqrt{\frac{\phi_2}{\phi_1}}$ (players' position interchanged) where it is mutually beneficial only for player 2 to make a transfer.

\section{Proof of Theorem \ref{theorem5}}\label{appndx:proof5}
For the specific parameters of this case we calculate the expected payoff of $T$ as a function of his own budget allocation across the physical ($\tau_1$) and social ($\tau_2$) battles following Eq. \ref{Eq12}. Without loss of generality we will assume that the budget relation of the opponents allocated for the physical fight is $\tau_1<d_1 \Leftrightarrow 1<\frac{d_1}{\tau_1}$, and for the cyber-social fight $d_2<\tau_2 \Leftrightarrow \frac{d_2}{\tau_2}<1$. Therefore the payoff will be:
\begin{equation}\small
    \begin{split}
        \Psi^{T}(\tau_1)=\frac{\phi_1\tau_1}{2d_1}+\phi_2 - \phi_2 \cdot \frac{d_2}{2\tau_2}
        \stackrel{\tau_2=\tau-\tau_1}{=}\\
        \frac{\phi_1\tau_1}{2d_1}+\phi_2 - \phi_2 \cdot \frac{d_2}{2(\tau-\tau_1)}
    \end{split}
\end{equation}
To calculate the optimal budget allocation for the TO:
\begin{equation}\small
\begin{split}
    \frac{\partial\Psi^{T}(\tau_1)}{\partial \tau_1}= \frac{\phi_1}{2d_1}-\frac{\phi_2d_2}{2(\tau-\tau_1)^2}=0 \Leftrightarrow ... \Leftrightarrow\\
    \tau-\tau_1 = \pm \sqrt{\frac{\phi_2 d_1 d_2}{\phi_1}}
\end{split}
\end{equation}
Since by definition $\tau-\tau_1>0$, and $\phi_1$,$\phi_2$,$d_1$,$d_2$ are positive values $\tau^{*}_1=\tau-\sqrt{\frac{\phi_2 d_1 d_2}{\phi_1}}$, which is a maximum for $\Psi^{T}(\tau_1)$ as $\frac{\partial^2\Psi^{T}(\tau_1)}{\partial \tau_1^2}=-\frac{\phi_2d_2(\tau-\tau_1)}{(\tau-\tau_1)^4}<0$.
Given that, without loss of generality we assume that the ESA (player 1) transfers budget equal to $r$ to the ICT agency (player 2), while their initial budget is $c_1, c_2$ respectively. Then the ESA payoff (Eq. \ref{Eq12}) is: 
\begin{equation}\small
    \Psi^{1}(r)=\phi_1- \phi_1 \frac{\tau^{*}_1(r)}{2(c_1-r)} = \phi_1 - \phi_1\frac{\tau-\sqrt{\frac{\phi_2(c_1-r)(c_2+r)}{\phi_1}}}{2(c_1-r)}
\end{equation}
Finding the first derivative yields:
\begin{equation}\small
\frac{\partial\Psi^{1}(r)}{\partial r}= -\frac{\phi_1 \tau}{2(c_1-r)^2}+\frac{\sqrt{\phi_1 \phi_2}}{4}\frac{(c_1+c_2)}{\sqrt{c_2+r}\cdot(c_1-r)^\frac{3}{2}}
\end{equation}
The optimal transfer $r^{*1\rightarrow{}2}$ is:
\begin{equation*}\small
\begin{split}
    \frac{\partial\Psi^{1}(r)}{\partial r}=0 \Leftrightarrow \frac{c_2+r}{c_1-r}=\frac{\phi_2\cdot(c_1+c_2)^2}{4\phi_1\cdot \tau^2} \Leftrightarrow 
    r = \frac{\frac{\phi_2\cdot(c_1+c_2)^2}{4\phi_1\cdot \tau^2}\cdot c_1 -c_2}{1+\frac{\phi_2\cdot(c_1+c_2)^2}{4\phi_1\cdot \tau^2}},\\ r^{*1\rightarrow{}2} = \frac{\xi\cdot c_1 -c_2}{1+\xi},~~ \xi=\frac{\phi_2\cdot(c_1+c_2)^2}{4\phi_1\cdot \tau^2}
\end{split}
\end{equation*}
Since $r\in[0,c_1)$, $\frac{\partial\Psi^{1}(r)}{\partial r}>0$ if $r<r^{*1\rightarrow{}2}$, and $\frac{\partial\Psi^{1}(r)}{\partial r}<0$ if $r>r^{*1\rightarrow{}2}$, $r^{*1\rightarrow{}2}$ is a maximum for $\Psi^{1}(r)$.
Thus, if $\frac{\partial\Psi^{1}(r)}{\partial r}\bigg|_{r=0}>0$ a sufficiently small positive transfer to player 2 will also benefit player 1. This holds iff 
\begin{equation*}\small
\begin{split}
-\frac{\phi_1 \tau}{2c_1^2}+\frac{\sqrt{\phi_1 \phi_2}}{2}\frac{(c_1+c_2)}{2\sqrt{c_2}\cdot c_1^\frac{3}{2}}>0 \Leftrightarrow ... \Leftrightarrow \\
\frac{c_1+c_2}{2\tau}>\sqrt{\frac{\phi_1 c_2}{\phi_2 c_1}}
\end{split}
\end{equation*}
which is a necessary condition for the existence of a budget transfer from player 1 to player 2 that is mutually beneficial. To ensure that player 2 is also benefited we can check: 
\begin{equation*}\small
\begin{split}
    \Psi^{2}(r)=\phi_2 \frac{c_2+r}{2\sqrt{\frac{\phi_2(c_1-r)(c_2+r)}{\phi_1}}},\\
    \frac{\partial\Psi^{2}(r)}{\partial r}= \frac{1}{4}\sqrt{\phi_1\phi_2}\frac{c_1+c_2}{\sqrt{c_2+r}\cdot(c_1-r)^\frac{3}{2}}>0,~~\forall r\in[0,c_1)
\end{split}
\end{equation*}
Therefore, the ICT agency always welcomes a positive transfer from the ESA in this case. Transfers that violate the conditions of this case are not beneficial for the ESA either (Appendix \ref{appndx:proof6}).

Regarding the opposite transfer ($r^{*2\rightarrow{}1}$) if we assume that a quantity $r$ is transferred from the ICT agency (player 2) to the ESA then the payoff of player 2 is given by:
\begin{equation}\small
\begin{split}
    \Psi^{2}(r)=\frac{\phi_2}{2} \frac{c_2-r}{\tau^{*}_2(r)} = \frac{\phi_2}{2} \frac{c_2-r}{\tau-\tau^{*}_1(r)}  = ... = \\ 
    \frac{\phi_2}{2} \frac{c_2-r}{\sqrt{\frac{\phi_2(c_1+r)(c_2-r)}{\phi_1}}}=\frac{\sqrt{\phi_1\phi_2}\cdot\sqrt{c_2-r}}{2\cdot\sqrt{c_1+r}}
\end{split}
\end{equation}
The first derivative:
\begin{equation}\small
\begin{split}
\frac{\partial\Psi^{2}(r)}{\partial r}= ... = -\frac{\sqrt{\phi_1\phi_2}}{4}\cdot \frac{c_1+c_2}{\sqrt{(c_2-r)}(c_1+r)^\frac{3}{2}}<0
\end{split}
\end{equation}
 $\forall~r \in [0,c_2)$.
Thus, $\Psi^{2}(r)$ is decreasing in this case and $r^{*2\rightarrow{}1}=0$ always since no transfer is beneficial for the ICT agency.

The same analysis can be followed for the case where the budget relation of the opponents allocated for the physical fight is $\tau_1>d_1$, and for the cyber-social fight $d_2>\tau_2$ (players' position interchanged). Then, it is mutually beneficial only for the ICT agency to make a transfer and for the ESA to accept it.

\section{Proof of Theorem \ref{theorem6}}\label{appndx:proof6}
We first consider the second SCDG stage. By Eq. \ref{Eq12}, $U^T_i=\phi_i\frac{\tau_i}{2d_i}$ for $\tau_i\leq d_i$ and $U^T_i=\phi_i-\phi_i\frac{d_i}{2\tau_i}$ for $\tau_i\geq d_i$. Its derivative with respect to $\tau_i$,
\begin{equation*}\small
m_i(\tau_i)=\begin{cases}\frac{\phi_i}{2d_i}, & \tau_i\leq d_i\\ \frac{\phi_id_i}{2\tau_i^2}, & \tau_i\geq d_i\end{cases}
\end{equation*}
is continuous and non-increasing, so $U^T_1+U^T_2$ is concave in the split and the first order conditions are sufficient: $m_1(\tau_1)=m_2(\tau_2)$ if $\tau_1,\tau_2>0$, and $m_2(\tau)\geq m_1(0)$ if $\tau_1=0$. Since $\frac{\phi_1}{d_1}\leq\frac{\phi_2}{d_2}$ we have $\delta\geq d_2$, and $m_2(\tau_2)\geq\frac{\phi_1}{2d_1}=m_1(0)$ iff $\tau_2\leq\delta$. This gives the first line of Eq. \ref{EqTO}. If $\delta<\tau\leq\delta+d_1$, then $\tau_2=\delta$ and $\tau_1=\tau-\delta\leq d_1$ give $m_1=m_2=\frac{\phi_1}{2d_1}$. If $\tau>\delta+d_1$, equal derivatives require $\tau_1>d_1$ and $\tau_2>\delta$, i.e., $\frac{\phi_1d_1}{\tau_1^2}=\frac{\phi_2d_2}{\tau_2^2}$, which gives $\tau_2=\sigma\tau_1$. The maximizer can fail to be unique only if both derivatives are constant around it, i.e., only if $\frac{\phi_1}{d_1}=\frac{\phi_2}{d_2}$ and $\tau<d_1+d_2$. In that case all splits with $\tau_i\leq d_i$ give the same payoff.

We will now focus on the first SCDG stage. Let $\kappa=\frac{\phi_1}{\phi_1+\phi_2}$, $\gamma=\frac{\phi_2}{\phi_1}$, $s=\frac{\tau}{C}$, $y=\frac{d_1}{C}$ and $h(y)=\frac{\delta}{C}=\sqrt{\gamma y(1-y)}$. Only the net transfer matters, and $\frac{\phi_1}{d_1}<\frac{\phi_2}{d_2}$ holds iff $y>\kappa$. The assumption $\frac{\phi_1}{c_1}<\frac{\phi_2}{c_2}$ means that the initial value $y_0=\frac{c_1}{C}$ exceeds $\kappa$, and $y=\kappa$ is reached at $r=\bar{r}$.

\textit{Step 1: payoffs.} For $y>\kappa$, Eq. \ref{EqTO} and Eq. \ref{Eq12} give
\begin{equation}\label{EqPsi1}\small
\frac{\Psi^1}{\phi_1}=\begin{cases}1, & s\leq h\\ 1-\frac{s-h}{2y}, & h<s\leq h+y\\ \frac{y+h}{2s}, & s>h+y\end{cases}
\end{equation}
\begin{equation*}\small
\Psi^2=\begin{cases}U^2(\tau,d_2), & s\leq h\\ \frac{\sqrt{\phi_1\phi_2}}{2}\sqrt{\frac{d_2}{d_1}}, & h<s\leq h+y\\ \frac{\phi_2d_2+\sqrt{\phi_1\phi_2d_1d_2}}{2\tau}, & s>h+y\end{cases}
\end{equation*}
Both are continuous for $y>\kappa$. The payoffs for $y<\kappa$ follow by exchanging the indices.

\textit{Step 2: $\Psi^2$ is strictly increasing in $d_2$ while $y>\kappa$.} This is clear for the first two lines. The third line is strictly concave in $d_2$ with maximum at $d_2=\frac{C}{2}(1+\sqrt{1-\kappa})$, by the computation of Appendix \ref{appndx:proof4} with the indices exchanged. Since $d_2<(1-\kappa)C$ and $z^2\leq\frac{1+z}{2}$ for $0\leq z\leq1$ (with $z=\sqrt{1-\kappa}$), $d_2$ lies below this maximum. By symmetry, $\Psi^1$ is strictly increasing in $d_1$ while $y<\kappa$. For $\tau\geq C$ payoffs are continuous at $y=\kappa$. For $\tau<C$, as $y$ increases through $\kappa$, $\tau_2^{*}$ jumps up from $\max\{0,\tau-d_1\}$ to $\min\{\tau,d_2\}$, and $\Psi^1$ jumps up.

\textit{Step 3: $r^{*2\rightarrow{}1}=0$.} Suppose that $r^{2\rightarrow{}1}>0$ in an SPNE. If $y>\kappa$, the ICT agency gains by lowering it slightly (Step 2). If $y\leq\kappa$ then $y<y_0$ and hence $r^{1\rightarrow{}2}>0$. If $y<\kappa$, or $y=\kappa$ and $\tau\geq C$, the ESA gains by lowering $r^{1\rightarrow{}2}$ slightly. If $y=\kappa$ and $\tau<C$, then either $\tau_2<\min\{\tau,d_2\}$ and the ESA gains by lowering $r^{1\rightarrow{}2}$ slightly, or $\tau_2=\min\{\tau,d_2\}$ and the ICT agency gains by lowering $r^{2\rightarrow{}1}$ slightly, as both deviations move the attack to the other game. Conversely, if $r^{2\rightarrow{}1}=0$ and $y>\kappa$ (or $y=\kappa$ with $\tau_2=\min\{\tau,d_2\}$), a transfer by the ICT agency raises $y$, lowers $d_2$, and lowers $\Psi^2$. The ESA never chooses $y<\kappa$ (Step 2). Hence the SPNE transfers are $r^{2\rightarrow{}1}=0$ together with the maximizers of $\Psi^1$ over $\kappa\leq y\leq y_0$. If $y_0=\kappa$ the same deviations exclude any transfer.

\textit{Step 4: maximizing $\Psi^1$ over $[\kappa,y_0]$.} The function $h(y)+y$ is concave, equals $1$ at $y=\kappa$ and $y=1$, and has its maximum $M=\frac{1}{2}(1+\sqrt{1+\gamma})$ at $y_B=\frac{1+\sqrt{\kappa}}{2}$. The second line of Eq. \ref{EqPsi1} is $1+\frac{q}{2}$ with $q=\frac{h-s}{y}=\sqrt{\gamma(u-1)}-su$, $u=\frac{1}{y}$. It is concave in $u$ with maximum at $u=1+\xi$, i.e., at $y_C=\frac{1}{1+\xi}$.

(i) If $s\geq M$, the third line of Eq. \ref{EqPsi1} applies throughout (it agrees with the second line where $s=h+y$). It is concave with maximum at $y_B>\kappa$, so the maximizer is $\min\{y_0,y_B\}$.

\looseness=-1 (ii) Let $\frac{1}{2}\max\{1,\sqrt{\gamma}\}<s<M$. Then $\xi<\min\{1,\gamma\}$, and $s>\frac{\sqrt{\gamma}}{2}\geq h(y)$ everywhere, so the first line never applies. Also $y_C>\kappa$ because $\xi<\gamma$. With $z=\sqrt{\xi}=\frac{\sqrt{\gamma}}{2s}$, the condition $s\leq h(y_C)+y_C$ reads $\sqrt{\gamma}z^2+2z-\sqrt{\gamma}\geq0$, which is equivalent to $s\leq M$. Thus $y_C$ lies where the second line applies. If $s>1$, the third line applies on two outer intervals. It is increasing on the left one, which lies to the left of $y_B$, and decreasing on the right one. Hence $\Psi^1$ increases up to $y_C$ and decreases after it, and the maximizer is $\min\{y_0,y_C\}$.

\looseness=-1 (iii)-(iv) Now $s\leq\frac{1}{2}\max\{1,\sqrt{\gamma}\}$. Where $s\leq h(y)$, $\Psi^1=\phi_1$ is the largest possible payoff. If $\gamma\leq1$, $h$ is decreasing on $[\kappa,1]$ with $h(\kappa)=1-\kappa$, so for $1-\kappa\leq s\leq\frac{1}{2}$ we have $s>h(y)$ for all $y>\kappa$. Moreover $y_C\leq\kappa$, since $\xi\geq\gamma$, so $\Psi^1$ is decreasing for $y>\kappa$. Its supremum is the limit at $y=\kappa$, which is attained iff the TO breaks the tie with $\tau_2^{*}=\min\{\tau,d_2\}$. This is (iii). In all other cases the set $\{y>\kappa:h(y)\geq s\}$ is a non-empty interval on which $\Psi^1=\phi_1$, and to the left of it $\Psi^1$ is increasing: the second line because $q<0$ is quasi-concave and reaches $0$ at the interval, and the third line (only if $s>1$) as in (ii). If the interval contains points $y\leq y_0$ these are the maximizers, together with $y=\kappa$ if $h(\kappa)\geq s$ and $\tau_2^{*}=\tau$. Otherwise the maximizer is $y_0$. This is (iv).

\textit{Relation to Theorems \ref{theorem4} and \ref{theorem5}.} The maximizers of (i) and (ii) are the transfers of these theorems. Conversely, if the conditions of Theorem \ref{theorem4} hold at its transfer and $r^{*}>0$, then $y=y_B$ and $s>h(y_B)+y_B=M$, which is (i). If the conditions of Theorem \ref{theorem5} hold at its transfer and $r^{*}>0$, then $y=y_C>\kappa$ and $h(y_C)<s<h(y_C)+y_C$, which give (ii). If $r^{*}=0$ in Theorem \ref{theorem4} and $y_0>\kappa$, then $y_0\leq y_B$. As $h+y$ is increasing up to $y_B$, the third line applies and increases up to $y_0$. The case $y_0\leq\kappa$ follows by exchanging the indices. If $r^{*}=0$ in Theorem \ref{theorem5}, then $\kappa<y_0\leq y_C$ and, as in (ii) and (iv), $\Psi^1$ increases up to $y_0$. Thus no transfer, including those that lead outside the case, improves the donor's payoff.

Finally, the donor maximizes its payoff over a set that includes $r=0$, and the receiver's payoff is strictly increasing in $r$ (Step 2). Thus, no SC entity is worse off.

\bibliographystyle{IEEEtran}
\bibliography{IEEEabrv,references}

\end{document}